\documentclass[final,twoside,11pt]{entics} 
\usepackage{enticsmacro}
\usepackage{graphicx}
\usepackage{url}
\usepackage{amsmath, amssymb}
\usepackage{proof}
\usepackage{hyperref, cleveref}
\usepackage{tikz}
\usetikzlibrary{cd}
\usetikzlibrary{positioning}
\tikzset{%
  symbol/.style={%
    draw=none,
    every to/.append style={%
      edge node={node [sloped, allow upside down, auto=false]{$#1$}}}
  }
}

\newcommand{\blank}{{-}}
\newcommand{\defeq}{\mathrel{:=}}

\newcommand{\sets}{\mathbf{Set}}
\newcommand{\cat}{\mathbf{CAT}}
\newcommand{\cc}{\mathbf{C}}

\newcommand{\cats}{\mathcal{S}}
\newcommand{\catr}{\mathcal{R}}
\newcommand{\catc}{\mathcal{C}}

\newcommand{\laxto}{\xrightarrow{\mathrm{lax}}}
\newcommand{\EM}[1]{\hat{#1}}

\newcommand{\ob}[1]{\mathop{\mathrm{Ob}}{#1}}
\newcommand{\opposite}{\mathrm{op}}
\newcommand{\Endo}[1]{\mathop{\mathrm{Endo}}({#1})}

\newcommand{\idmorph}{\mathrm{id}}
\newcommand{\idfunc}{\mathrm{Id}}

\newcommand{\twocell}{\Rightarrow}
\newcommand{\nattr}{\Rightarrow}
\newcommand{\xto}[1]{\xrightarrow{#1}}

\newcommand{\EMfree}{\hat{J}}
\newcommand{\EMforget}{\hat{E}}
\newcommand{\EMunit}{\hat{\eta}}
\newcommand{\EMcounit}{\hat{\varepsilon}}

\newcommand{\gu}{\overline{\eta}}
\newcommand{\gcu}{\overline{\varepsilon}}

\newcommand{\Term}[3]{\mathrm{Term}_{#1}({#2}, {#3})}
\newcommand{\op}{\mathrm{op}}
\newcommand{\pure}{\mathop{\mathbf{e}}}
\newcommand{\doop}{\mathop{\mathbf{do}}}

\newcommand{\oprtyp}[3]{{#1} \rightsquigarrow {#2}; {#3}}
\newcommand{\interp}[2]{{|{#2}|}^{#1}}

\newcommand{\gat}[1]{\mathcal{#1}}
\newcommand{\GATMod}[3]{\mathrm{Mod}_{#2}^{#1}({#3})}

\newcommand{\free}[2]{F_{#2}^{#1}}
\newcommand{\freeex}[1]{\overline{#1}}
\newcommand{\forget}[2]{U_{#2}^{#1}}

\newcommand{\num}{\mathtt{int}}
\newcommand{\unit}{\mathbf{1}}

\newcommand{\unitval}{\star}

\newcommand{\recv}{\mathord{recv}}
\newcommand{\send}{\mathord{send}}
\newcommand{\recvint}{\mathrm{recvint}}
\newcommand{\sendint}{\mathrm{sendint}}
\newcommand{\lookupint}{\mathrm{lookupint}}
\newcommand{\updateint}{\mathrm{updateint}}

\newcommand{\cateff}{\textit{CatEff}}

\newcommand{\ret}{\mathop{\mathbf{val}}\nolimits}
\newcommand{\lt}{\mathop{\mathbf{let}}}
\newcommand{\bind}{\leftarrow}
\newcommand{\lin}{\mathop{\mathbf{in}}\nolimits}
\newcommand{\letin}[3]{
  \lt {#1} \bind {#2} \lin {#3}
}
\newcommand{\handle}[2]{\mathop{\mathbf{handle}} {#1} \mathop{\mathbf{with}} {#2}}
\newcommand{\inl}{\mathop{\mathbf{inl}}}
\newcommand{\inr}{\mathop{\mathbf{inr}}}
\newcommand{\pair}[2]{\langle {#1}, {#2} \rangle}
\newcommand{\proj}[4]{\mathop{\mathbf{proj}} {#1} \mathop{\mathbf{as}} {\langle {#2},{#3} \rangle} .{#4}}

\newcommand{\mc}{\mathop{\mathbf{match}}}
\newcommand{\match}[5]{\mc{#1} \{ {#2}.{#3}; {#4}.{#5} \}}

\newcommand{\opr}{\mathop{\mathrm{op}}}

\newcommand{\handlerto}{\Rightarrow}

\newcommand{\rto}{\rightarrow}
\newcommand{\ctx}{\mathcal{E}}
\newcommand{\ctxf}{\mathcal{F}}
\newcommand{\ctxto}{\rightsquigarrow}

\newcommand{\update}[2]{\mathrm{update}^{#1}_{#2}}
\newcommand{\lookup}[1]{\mathrm{lookup}_{#1}}
\newcommand{\upd}{\mathrm{update}}
\newcommand{\lkp}{\mathrm{lookup}}

\newcommand{\denote}[1]{[\![{#1}]\!]}

\newcommand{\inj}{\mathop{\mathrm{in}}\nolimits}
\newcommand{\apx}{\lhd}

\def\conf{MFPS 2022} 	
\volume{1}			
\def\volu{1}			
\def\papno{15}			
\def\mydoi{{\normalshape \href{https://doi.org/10.46298/entics.10491}{10.46298/entics.10491}}}
\def\copynames{T. Sanada} 
\def\runtitle{Category-Graded Algebraic Theories and Effect Handlers}

\def\lastname{Sanada}

\begin{document}

\begin{frontmatter}
  \title{Category-Graded Algebraic Theories and Effect Handlers} 
  \author{Takahiro Sanada\thanksref{email}}	
  \address{Research Institute for Mathematical Sciences\\ Kyoto University\\ Kyoto, Japan}
  \thanks[email]{Email: \href{mailto:tsanada@kurims.kyoto-u.ac.jp} {\texttt{\normalshape tsanada@kurims.kyoto-u.ac.jp}}}
\begin{abstract}
  We provide an effect system \cateff{} based on
  a category-graded extension of algebraic theories
  that correspond to category-graded monads.
  \cateff{} has category-graded operations and handlers.
  Effects in \cateff{} are graded by morphisms of the grading category.
  Grading morphisms represent fine structures of effects such as dependencies or sorts of states.
  Handlers in \cateff{} are regarded as an implementation of category-graded effects.
  We define the notion of category-graded algebraic theory to give semantics of \cateff{}
  and prove soundness and adequacy.
  We also give an example using category-graded effects
  to express protocols for sending receiving typed data.
\end{abstract}
\end{frontmatter}

\begin{abstract}
    We provide an effect system \cateff{} based on
    a category-graded extension of algebraic theories
    that correspond to category-graded monads.
    Our effect system has category-graded operations and handlers.
    Effects are graded by morphisms of the grading category.
    Grading morphisms represent fine structures of effects
    such as dependencies or sorts of states.
    Category-graded handlers
    are regarded as an implementation of category-graded effects.
    We give an example using category-graded effects
    to represent computational effects such as sending and receiving typed data.
\end{abstract}

\section{Introduction}
\subsection{Background}
Moggi \cite{Moggi91} used \textit{monads} to capture computational effects.
Monads have a close relationship with \textit{algebraic theories} \cite{HP07}.
Algebraic effects \cite{PlPo01} are effects based on algebraic theories.
Handlers of algebraic effects \cite{PlPr09} provide clear ways to implement effects.
Algebraic effects and handlers are useful notions to make programs with effects.

There are several extensions of monads \cite{Atkey09,Katsumata14,OWE20,OrchardPM14}.
These variations of monads enable us to reason about computational effects in more detail.

\textit{Parameterised monads} \cite{Atkey09} are monads with parameters
which represent initial and terminal states of computational effects
such as change of type of state.
In an effect system based on parameterised monads,
each computational term is graded by an object of a parameter category $\cats$.
For example,
we can capture the feature of mutable state of mutable type.
To see this, let $\cats$ be a discrete category whose objects are $\num$ and $\unit$.
The parameter indicates the type of state.
We can construct computational terms,
$M$ with parameter $\num$ and $N$ with parameter $\unit$.
We can know that
the computations $M$ and $N$ have states of type $\num$ and $\unit$ respectively.
Let us consider lookup and update operations for this mutable state of mutable type.
There are two lookup operations and four update operations
$\lkp_{\num}(), \lkp_{\unit}()$,
$\upd_{\num \to \unit}(V)$,
$\upd_{\unit \to \num}(V)$,
$\upd_{\num \to \num}(V)$ and
$\upd_{\unit \to \unit}(V)$.
$\lkp_{\num}()$ is an operation that reads the state of type $\num$
and returns the value in it.
We can use $\lkp_{\num}()$ only when the type of the state is $\num$.
$\lkp_{\unit}()$ is similar.
$\upd_{\alpha \to \beta}(V)$ is an operation that writes the value $V$ in the state
changing the type of state from $\alpha$ to $\beta$.
The parameter category $\cats$ does not need to be a discrete category.
If $\cats$ has nontrivial morphisms,
the intuition is that morphisms mean subtyping relations between the types of the state.

\textit{Graded monads} \cite{Katsumata14,Mellies12} are monads graded by a partially ordered monoid.
Elements of the partially ordered monoid express quantity of effects
such as memory locations that effects affect.
Its formal theory was given in \cite{FKM16},
and its algebraic theories were given in \cite{Smirnov08,DorschMS19,Kura20}.
For example,
let $L = \{ l_1, l_2, \dots, l_n \}$ be a set of memory locations.
Then we can get a partially ordered monoid $2^L$
where the product $\cdot$ of $2^L$ is the union of sets $\cup$ and
the order $\le$ of $2^L$ is the inclusion $\subseteq$.
If a computation term $M$ is graded by $A \in 2^L$,
we can know that $M$ may access the memory locations contained in $A$.
The role of the order is weakening of a set of locations which may be accessed.
If a computation term $M$ is graded by $A \in 2^L$ and $A \le B$,
we can deduce that $M$ is also graded by $B$.
The intuition is that
computation that may access the locations in $A$
may access the locations in $B$ which is larger than $A$.
Let us consider memory lookup and update operations.
Let $\lkp_i()$ be an operation that reads location $l_i$ and returns its value
and $\upd_j(V)$ be an operation that writes $V$ on the location $l_j$
and returns the unit value.
$\lkp_i()$ and $\upd_j(V)$ are graded by $\{ l_i \}$ and $\{ l_j \}$ respectively.
If $M$ is graded by $A \in 2^L$,
$\mathop{\mathbf{let}} x \leftarrow \lkp_i() \mathop{\mathbf{in}} M$
and
$\mathop{\mathbf{let}} \_ \leftarrow \upd_j(V) \mathop{\mathbf{in}} M$
are graded by $\{ l_i \} \cdot A$ and $\{ l_j \} \cdot A$, respectively.

\textit{Category-graded monads} \cite{OWE20} are introduced
to unify parameterised and graded monads.
Graded monads are 2-category-graded monads with a single object.
Parameterised monads are category-graded monads with \textit{generalised units}.
Category-graded monads and the constructions of these Eilenberg-Moore and Kleisli categories
are a special case of lax functors and these two constructions are studied by Street \cite{Street72}.

\subsection{Overview}
In this paper, we provide
\begin{itemize}
\item category-graded extensions of algebraic theories,
\item a category-graded effect system \cateff{} with effect handlers
  based on category-graded algebraic theories, and
\item operational and denotational semantics of the effect system.
\end{itemize}

In category-graded algebraic theories,
terms are graded by morphisms of a grading category $\cats$.
In the effect system that corresponds to category-graded algebraic theories,
we will define a judgement $\Gamma \vdash_f M : A$ for computational term $M$
and a morphism $f$ in $\cats$.
This judgement means that
the computation $M$ will return a value of type $A$ under the environment $\Gamma$
and invoke effects graded by $f$.
The term $M$ will be denoted by a map $\denote{\Gamma} \to T_f \denote{A}$,
where $T$ is a category-graded monad.
Grading morphisms can express
a finer structure of computational effects
than
elements of monoids in ordinary graded monads
or parameters in parameterised monads can,
especially structures of dependency and sorts of state.

For instance,
let us consider the following morphisms
$f = (\unit \xto{\tau^{\unit}_{\num}} \num \xto{\send_{\num}} \num \xto{\recv^{\num}_{\num}} \num)$ and
$g = (\unit \xto{\recv^{\unit}_{\num}} \num \xto{\send_{\num}} \num)$
in a category $\cats$.
A computational term $M$ graded by $f$
is a computation that behaves as follows.
\begin{enumerate}
\item The type of the initial state is unit type $\unit$.
\item Some effects are invoked in $M$ and
  a value of type $\num$ is stored in the state.
  Thus, the type of state is changed from $\unit$ to $\num$.
  These effects are graded by $\tau^{\unit}_{\num}$,
  which means internal computation with a change of types of the state.
\item An effect sending the value of the state to another process is invoked.
  It is graded by $\send_{\num}$.
  The type of the state is not changed by the sending effect,
  so the domain and codomain of $\send_{\num}$ are the same.
\item An effect receiving graded by $\recv^{\num}_{\num}$ is invoked.
  It receives a value of type $\num$ and stores it in the state.
  In this case,
  the types of state before receiving and after receiving are the same,
  but in general, they may be changed.
\end{enumerate}
A computational term $N$ graded by $g$ is a computation that behaves as follows.
\begin{enumerate}
\item The type of the initial type is unit type $\unit$.
\item An effect graded by $\recv^{\unit}_{\num}$ is invoked.
  It receives a value of type $\num$ and stores it in the state.
  The type of state before receiving and the type of receiving value are different,
  so the type of state is changed from $\unit$ to $\num$.
\item An effect graded by $\send_{\num}$ is invoked.
  It sends a value of the state.
\end{enumerate}
Thanks to the grading morphisms,
we can know the transition of the type of the state,
and deduce that $M$ and $N$ can interact with each other and yield values.
We think of morphisms in $\cats$ as protocols of communication.

We can construct handlers of category-graded effects.
As ordinary handlers are
the morphisms induced by the universality of free models of algebraic theories,
category-graded handlers are
the morphisms induced by
the universality of free models of category-graded algebraic theories.
We can regard category-graded handlers
as an implementation of category-graded effects.
(Monoid-)graded monads without order and
parameterised monads whose parameterising category is discrete
are special category-graded monads,
so we can get handlers for effects corresponding to these monads automatically.

\textbf{Contents}.
In \Cref{preliminaries}, we introduce notations and review some categorical notions.
In \Cref{category-graded-algebraic-theory},
we define category-graded algebraic theory
and describe the free construction for the theory.
In \Cref{category-graded-effect-system},
we explain our effect system based on category-graded algebraic theory.
We call the effect system \cateff{}.
The effect system has handlers of category-graded effects.
In \Cref{operational-semantics} and \Cref{denotational-semantics},
we describe operational and denotational semantics of our effect system, respectively.
In \Cref{soundness-and-adequacy}, we show the soundness and adequacy of the semantics.
We argue the theory of the generalised units and its integration into \cateff{} in \Cref{generalised-units} and \Cref{correspondence-of-p-monads-to-cat-graded-monads-with-gu}.
\section{Category-Graded Monads} \label{preliminaries}
We assume that readers are familiar
with basic notions of category theory
such as monads \cite{MacLane71}.
Throughout this paper, we use the following notations.
\begin{itemize}
\item Let $\cat$ be the 2-category of all categories, functors and natural transformations.
\item Let $\sets$ be the category of all sets and maps.
\item For a category $\catc$,
  $\idfunc_{\catc}$ is the identity functor on $\catc$.
\item For a category $\catc$,
  $\Endo{\catc}$ is the full 2-subcategory of $\cat$
  whose 0-cell is only $\catc$.
\item For a category $\catc$ with finite and countable products and coproducts,
  an object $C$ of $\catc$,
  and a finite or countable set $X$,
  $C^X$ is the $X$-fold product of $C$, that is $\prod_{x \in X} C$ and
  $X \times C$ is the $X$-fold sum of $C$, that is $\sum_{x \in X} C$.
\end{itemize}

\subsection{Lax Functors and Category-Graded Monads}
Category-graded monads are introduced by Orchard et al. \cite{OWE20}.
In this section, we fix a category $\catc$ and a small category $\cats$.

\begin{definition}[Lax functor]
  Let $\cc$ be a 2-category.
  A \textit{lax functor} $F \colon \cats \laxto \cc$ is
  a tuple $F = (F, \eta^F, \mu^F)$
  where
  \begin{itemize}
  \item For each object $a$ of $\cats$,
    $Fa$ is a 0-cell of $\cc$.
  \item For each morphism $f \colon a \to b$ of $\cats$,
    $Ff \colon Fa \to Fb$ is a 1-cell of $\cc$.
  \item For each object $a$ of $\cats$,
    $\eta^F_a \colon \idmorph_{Fa} \twocell F \idmorph_a$ is a 2-cell of $\cc$.
  \item For each morphism $f \colon a \to b$ and $g \colon b \to c$ of $\cats$,
    $\mu^F_{g, f} \colon Fg \circ Ff \twocell F(g \circ f)$
    is a 2-cell of $\cc$
  \end{itemize}
  satisfying the following commutative diagrams:
  \[
  \begin{tikzcd}
    Ff
    \ar[r, Rightarrow, "Ff \eta^F_a"]
    \ar[d, Rightarrow, "\eta^F_b Ff"']
    \ar[rd, equal]
    & Ff \circ F\idmorph_a
    \ar[d, Rightarrow, "\mu^F_{f, \idmorph_a}"]
    \\
    F\idmorph_b \circ Ff
    \ar[r, Rightarrow, "\mu^F_{\idmorph_b, f}"']
    &
    Ff
  \end{tikzcd}
  \hspace{1cm}
  \begin{tikzcd}
    Fh \circ Fg \circ Ff
    \ar[r, Rightarrow, "Fh \mu^F_{g,f}"]
    \ar[d, Rightarrow, "\mu^F_{h,g} Ff"']
    & Fh \circ F(g \circ f)
    \ar[d, Rightarrow, "\mu^F_{h, g \circ f}"]
    \\
    F(h \circ g) Ff
    \ar[r, Rightarrow, "\mu^F_{h \circ g, f}"']
    & F(h \circ g \circ f)
  \end{tikzcd}
  \]
  We call
  $\eta^F$ and $\mu^F$
  \textit{unit} and \textit{multiplication} of $F$ respectively.
\end{definition}
We use string diagrams \cite{Selinger10} for diagrammatic reasoning.
In string diagram, a region of a diagram represents a 0-cell of a 2-category,
a string between two regions represents a 1-cell
from the 0-cell of the left regions to the 0-cell of the right region,
and a node on strings represents a 2-cells from the 1-cell of bottom strings to the 1-cell of top strings.
We can depict unit and multiplication,
and the axioms of lax functor
by string diagrams as follows:
\[
\eta^F_a = \begin{tikzpicture}[baseline=0.5cm, x=0.5cm, y=0.5cm]
\coordinate (eta) at (1.000000, 1.000000);
\coordinate (rb) at (2.000000, 0.000000);
\coordinate (rt) at (2.000000, 2.000000);
\coordinate (Fid) at (1.000000, 2.000000);
\coordinate (lt) at (0.000000, 2.000000);
\coordinate (lb) at (0.000000, 0.000000);
\fill[red!20] (lt) -- (lb)  -- (rb) -- (rt)-- cycle;
\draw (eta) -- (Fid);

\node[above] at (Fid) {$F \idmorph_a$};
\node[below] at (eta) {$\eta^F_a$};
\fill[black] (eta) circle [radius=0.1];
\end{tikzpicture}, \quad
\mu^F_{f,g} = \begin{tikzpicture}[baseline=0.5cm, x=0.5cm, y=0.5cm]
\coordinate (rt) at (4.000000, 2.000000);
\coordinate (rb) at (4.000000, 0.000000);
\coordinate (lb) at (0.000000, 0.000000);
\coordinate (lt) at (0.000000, 2.000000);
\coordinate (Fgf) at (2.000000, 2.000000);
\coordinate (mu) at (2.000000, 1.000000);
\coordinate (Fg) at (3.000000, 0.000000);
\coordinate (Ff) at (1.000000, 0.000000);
\fill[green!20] (rt) -- (rb)  -- (Fg) to[out=90, in=0] (mu) -- (Fgf)-- cycle;
\fill[blue!20] (Ff) to[out=90, in=180] (mu)  to[out=0, in=90] (Fg)-- cycle;
\fill[red!20] (lt) -- (lb)  -- (Ff) to[out=90, in=180] (mu) -- (Fgf)-- cycle;
\draw (mu) -- (Fgf);
\draw (Fg) to[out=90, in=0] (mu);
\draw (Ff) to[out=90, in=180] (mu);

\node[above] at (Fgf) {$F(g \circ f)$};
\node[below] at (Ff) {$Ff$};
\node[below] at (Fg) {$Fg$};
\node[below] at (mu) {$\mu^F_{f,g}$};

\fill[black] (mu) circle [radius = 0.1];
\end{tikzpicture}, \quad
\begin{tikzpicture}[baseline=0.5cm, x=0.5cm, y=0.5cm]
\coordinate (rb) at (4.000000, 0.000000);
\coordinate (rt) at (4.000000, 3.000000);
\coordinate (lt) at (0.000000, 3.000000);
\coordinate (lb) at (0.000000, 0.000000);
\coordinate (Ffb) at (3.000000, 0.000000);
\coordinate (Ff) at (2.000000, 3.000000);
\coordinate (mu) at (2.000000, 2.000000);
\coordinate (eta) at (1.000000, 1.000000);
\fill[blue!20] (rt) -- (rb)  -- (Ffb) to[out=90, in=0] (mu) -- (Ff)-- cycle;
\fill[red!20] (lt) -- (lb)  -- (Ffb) to[out=90, in=0] (mu) -- (Ff)-- cycle;
\draw (mu) -- (Ff);
\draw (Ffb) to[out=90, in=0] (mu);
\draw (eta) to[out=90, in=180] (mu);

\node[above] at (Ff) {$Ff$};
\node[below] at (Ffb) {$Ff$};
\fill[black] (eta) circle [radius=0.1];
\fill[black] (mu) circle [radius=0.1];
\end{tikzpicture}
=
\begin{tikzpicture}[baseline=0.5cm, x=0.5cm, y=0.5cm]
\coordinate (rb) at (2.000000, 0.000000);
\coordinate (rt) at (2.000000, 3.000000);
\coordinate (lt) at (0.000000, 3.000000);
\coordinate (lb) at (0.000000, 0.000000);
\coordinate (Ffb) at (1.000000, 0.000000);
\coordinate (Ff) at (1.000000, 3.000000);
\fill[red!20] (lt) -- (lb)  -- (Ffb) -- (Ff)-- cycle;
\fill[blue!20] (rt) -- (rb)  -- (Ffb) -- (Ff)-- cycle;
\draw (Ffb) -- (Ff);

\node[above] at (Ff) {$Ff$};
\node[below] at (Ffb) {$Ff$};
\end{tikzpicture}
=
\begin{tikzpicture}[baseline=0.5cm, x=0.5cm, y=0.5cm]
\coordinate (rb) at (4.000000, 0.000000);
\coordinate (rt) at (4.000000, 3.000000);
\coordinate (lt) at (0.000000, 3.000000);
\coordinate (lb) at (0.000000, 0.000000);
\coordinate (Ffb) at (1.000000, 0.000000);
\coordinate (Ff) at (2.000000, 3.000000);
\coordinate (mu) at (2.000000, 2.000000);
\coordinate (eta) at (3.000000, 1.000000);
\fill[red!20] (lt) -- (lb)  -- (Ffb) to[out=90, in=180] (mu) -- (Ff)-- cycle;
\fill[blue!20] (rt) -- (rb)  -- (Ffb) to[out=90, in=180] (mu) -- (Ff)-- cycle;
\draw (mu) -- (Ff);
\draw (Ffb) to[out=90, in=180] (mu);
\draw (eta) to[out=90, in=0] (mu);

\node[above] at (Ff) {$Ff$};
\node[below] at (Ffb) {$Ff$};
\fill[black] (eta) circle [radius=0.1];
\fill[black] (mu) circle [radius=0.1];
\end{tikzpicture},
\]
\[
\begin{tikzpicture}[baseline=0.5cm, x=0.5cm, y=0.5cm]
\coordinate (rb) at (6.000000, 0.000000);
\coordinate (rt) at (6.000000, 3.000000);
\coordinate (lb) at (0.000000, 0.000000);
\coordinate (lt) at (0.000000, 3.000000);
\coordinate (Fhgf) at (3.000000, 3.000000);
\coordinate (m1) at (3.000000, 2.000000);
\coordinate (m0) at (2.000000, 1.000000);
\coordinate (Fh) at (5.000000, 0.000000);
\coordinate (Fg) at (3.000000, 0.000000);
\coordinate (Ff) at (1.000000, 0.000000);
\fill[yellow!20] (Fh) to[out=90, in=0] (m1)  -- (Fhgf) -- (rt) -- (rb)-- cycle;
\fill[green!20] (Fg) to[out=90, in=0] (m0)  to[out=90, in=180] (m1) to[out=0, in=90] (Fh)-- cycle;
\fill[blue!20] (Ff) to[out=90, in=180] (m0)  to[out=0, in=90] (Fg)-- cycle;
\fill[red!20] (lt) -- (lb)  -- (Ff) to[out=90, in=180] (m0) to[out=90, in=180] (m1) -- (Fhgf)-- cycle;
\draw (Fh) to[out=90, in=0] (m1);
\draw (Fg) to[out=90, in=0] (m0);
\draw (Ff) to[out=90, in=180] (m0);
\draw (m0) to[out=90, in=180] (m1);
\draw (m1) -- (Fhgf);

\node[above] at (Fhgf) {$F(h \circ g \circ f)$};
\node[below] at (Ff) {$Ff$};
\node[below] at (Fg) {$Fg$};
\node[below] at (Fh) {$Fh$};

\fill[black] (m0) circle [radius=0.1];
\fill[black] (m1) circle [radius=0.1];
\end{tikzpicture}
=
\begin{tikzpicture}[baseline=0.5cm, x=0.5cm, y=0.5cm]
\coordinate (rb) at (6.000000, 0.000000);
\coordinate (rt) at (6.000000, 3.000000);
\coordinate (lb) at (0.000000, 0.000000);
\coordinate (lt) at (0.000000, 3.000000);
\coordinate (Fhgf) at (3.000000, 3.000000);
\coordinate (m1) at (3.000000, 2.000000);
\coordinate (m0) at (4.000000, 1.000000);
\coordinate (Fh) at (5.000000, 0.000000);
\coordinate (Fg) at (3.000000, 0.000000);
\coordinate (Ff) at (1.000000, 0.000000);
\fill[red!20] (lt) -- (lb)  -- (Ff) to[out=90, in=180] (m1) -- (Fhgf)-- cycle;
\fill[blue!20] (Fg) to[out=90, in=180] (m0)  to[out=90, in=0] (m1) to[out=180, in=90] (Ff)-- cycle;
\fill[green!20] (Fg) to[out=90, in=180] (m0)  to[out=0, in=90] (Fh)-- cycle;
\fill[yellow!20] (rt) -- (rb)  -- (Fh) to[out=90, in=0] (m0) to[out=90, in=0] (m1) -- (Fhgf)-- cycle;
\draw (m0) to[out=90, in=0] (m1);
\draw (m1) -- (Fhgf);
\draw (Ff) to[out=90, in=180] (m1);
\draw (Fg) to[out=90, in=180] (m0);
\draw (Fh) to[out=90, in=0] (m0);

\node[above] at (Fhgf) {$F(h \circ g \circ f)$};
\node[below] at (Ff) {$Ff$};
\node[below] at (Fg) {$Fg$};
\node[below] at (Fh) {$Fh$};

\fill[black] (m0) circle [radius=0.1];
\fill[black] (m1) circle [radius=0.1];
\end{tikzpicture}.
\]

\begin{definition}[Category-graded monads]
  A \textit{category-graded monad} (or an \textit{$\cats$-graded monad}) on $\catc$
  is a lax functor $\cats^{\opposite} \laxto \Endo{\catc}$.
  That is, an $\cats$-graded monad consists of mapping of objects and morphisms
  $T \colon \cats^{\opposite} \to \Endo{\catc}$ and
  families of natural transformations
  $\eta_a \colon \idfunc_{\catc} \nattr T_{\idmorph_a}$
  for $a \in \cats$
  and
  $\mu_{f,g} \colon T_f T_g \nattr T_{f;g}$
  for $f \colon a \to b$ and $g \colon b \to c$ in $\cats$
  that make the following diagrams commute.
  \[
  \begin{tikzcd}
    T_f
    \ar[r, Rightarrow, "\eta_a T_f"]
    \ar[d, Rightarrow, "T_f \eta_b"']
    \ar[rd, equal]
    & T_{\idmorph_a} T_f
    \ar[d, Rightarrow, "\mu_{\idmorph_a, f}"] \\
    T_f T_{\idmorph_b}
    \ar[r, Rightarrow, "\mu_{f, \idmorph_b}"']
    & T_f
  \end{tikzcd} \hspace{1cm}
  \begin{tikzcd}
    T_f T_g T_h
    \ar[r, Rightarrow, "T_f \mu_{g,h}"]
    \ar[d, Rightarrow, "\mu_{f,g} T_h"']
    & T_f T_{g; h}
    \ar[d, Rightarrow, "\mu_{f, g; h}"] \\
    T_{f;g} T_h
    \ar[r, Rightarrow, "\mu_{f;g, h}"']
    & T_{f; g; h}
  \end{tikzcd}
  \]
\end{definition}

If $\cats$ is the trivial category,
that is the category with a single object
and the identity morphism on it,
$\cats$-graded monad is a usual monad \cite{benabou67}.
If $\cats$ is a category with single object
and endomorphisms on it,
the endomorphisms form a monoid and
$\cats$-graded monad is a (monoid-)graded monad without order.
To consider parameterised monads as category-graded monads introduced in \cite{OWE20},
we have to introduce \textit{generalised units}, see \cite{OWE20}.

\subsection{Eilenberg-Moore Construction on Lax Functors}
\label{Eilenberg-Moore-construction-of-lax-functors}
According to \cite{Street72}, there are two functors obtained from a lax functor.
The two constructions correspond to
Eilenberg-Moore and Kleisli construction on a monad, respectively.
In this subsection, we review the Eilenberg-Moore construction on lax functors.

Let $\cats$ be a small category and $F \colon \cats \to \cat$ be a lax functor.
The Eilenberg-Moore construction gives a functor $\EM{F} \colon \cats \to \cat$.
\begin{definition}
  For a lax functor $F = (F, \eta^F, \mu^F) \colon \cats \to \cat$,
  the functor $\EM{F} \colon \cats \to \cat$ is defined as follows:
  \begin{itemize}
  \item For an object $a \in \ob{\cats}$,
    the category $\EM{F}a$ is defined as follows.
    \begin{itemize}
    \item Objects are pairs $(A,\xi)$ where
      $A$ is a map which assigns to each morphism
      $f \colon a \to b$ of $\cats$ an object $A_f \in \ob{Fb}$
      and $\xi$ is a family of morphisms
      ${\left\{
        \xi_{f,g} \colon F_f A_g \to A_{f \circ g}
        \right\}}_{f, g}$
      such that the following diagrams commute:   
      \[
      \begin{tikzcd}
        A_f
        \ar[r, "\eta^F_b A_f"]
        \ar[rd, equal]
        & F_{\idmorph_b} A_f
        \ar[d, "\xi_{\idmorph_b, f}"] \\
        & A_f
      \end{tikzcd}
      \quad
      \begin{tikzcd}
        F_h F_g A_f
        \ar[r, "F_h \xi_{g,f}"]
        \ar[d, "\mu^F_{h,g} A_f"']
        & F_h A_{g \circ f}
        \ar[d, "\xi_{h, g \circ f}"] \\
        F_{h \circ g} A_f
        \ar[r, "\xi_{h \circ g, f}"']
        & A_{h \circ g \circ f}
      \end{tikzcd}
      \]
      for each $f \colon a \to b$, $g \colon b \to c$, $h \colon c \to d$ in $\cats$.
    \item Morphisms are $\alpha \colon (A,h) \to (A', \xi')$ where
      $\alpha$ is a family of morphisms
      ${\left\{
        \alpha_f \colon A_f \to A'_f
        \right\}}_{f}$
      such that
      the following diagram commutes:
      \[
      \begin{tikzcd}
        F_g A_f
        \ar[r, "F_g \alpha_f"]
        \ar[d, "\xi_{g,f}"']
        & F_g A'_f
        \ar[d, "\xi'_{g,f}"] \\
        A_{g \circ f}
        \ar[r, "\alpha_{g \circ f}"']
        & A'_{g \circ f}
      \end{tikzcd}
      \]
      for each $f \colon a \to b$, $g \colon b \to c$.
    \end{itemize}
  \item For a morphism $f \colon a \to b$,
    a functor $\EM{F}f \colon \EM{F}a \to \EM{F}b$ is defined as follows:
    \begin{itemize}
    \item $\EM{F}f$ assigns
      an object $(B, \zeta) \defeq (\EM{F}f)(A,\xi)$ of $\EM{F}b$
      to an object $(A, \xi)$ of $\EM{F}a$
      where
      $B_g = A_{g \circ f}$
      and
      $\zeta_{h,g} = \xi_{h, g \circ f}$
      for $g \colon b \to c$, $h \colon c \to d$.
    \item For a morphism $\alpha \colon (A, \xi) \to (A', \xi')$ of $\EM{F}a$,
      a morphism
      \[ (\EM{F}f)\alpha \colon (\EM{F}f)(A, \xi) \to (\EM{F}f)(A', \xi') \]
      is defined as
      ${((\EM{F}f)\alpha)}_g = \alpha_{g \circ f}$ for $g \colon b \to c$ in $\cats$.
    \end{itemize}
  \end{itemize}
\end{definition}

Since $\cats$-graded monad
$T \colon \cats^{\opposite} \to \Endo{\catc}$
is a special case of lax functor $\cats^{\opposite} \to \cat$,
the category of Eilenberg-Moore algebras of $\cats$-graded monad $T$
is obtained by the above construction.

Next, we describe an adjunction between $Fa$ and $\EM{F}a$ for $a \in \ob{\cats}$.
A functor $\EMfree_a \colon Fa \to \EM{F}a$ is defined as follows.
\begin{itemize}
\item For an object $X \in \ob{Fa}$,
  $\EMfree_a X \defeq (A, \xi) \in \ob{\EM{F}a}$
  where
  $A_f = F_f X$ and
  $\xi_{g,f} = \mu^F_{g,f,X} \colon F_g A_f \to A_{g \circ f}$
  for $f \colon a \to b$, $g \colon b \to c$ in $\cats$.
\item For a morphism $x \colon X \to Y$ in $Fa$,
  $\EMfree_a x \colon \EMfree_a X \to \EMfree_a Y$ is a family of
  $ {\EMfree_a x}_{f} \defeq F_f x$.
\end{itemize}
A functor $\EMforget_a \colon \EM{F}a \to Fa$ is defined as follows.
\begin{itemize}
\item For an object $(A, \xi) \in \ob{\EM{F}a}$,
  $\EMforget_a (A, \xi) \defeq A_{\idmorph_a} \in Fa$.
\item For a morphism $\alpha \colon (A, \xi) \to (A', \xi')$ in $\EM{F}a$,
  $\EMforget_a \alpha = \alpha_{\idmorph_a}$.
\end{itemize}
We will show that $\EMfree$ is a left adjoint to $\EMforget$.
To do so, we construct a unit and a counit of the adjunction.
The unit $\EMunit_a \colon \idfunc_{Fa} \nattr \EMforget_a \EMfree_a$ is
a natural transformation whose components are
$\EMunit_{a,X} \defeq \eta^F_X \colon X \to F_{\idmorph_a}X$.
The counit $\EMcounit_a \colon \EMfree_a \EMforget_a \nattr \idfunc_{\EM{F}a}$ is
a natural transformation whose components are
\[ \EMcounit_{a, (A,\xi)} \defeq \xi_{(\blank), \idmorph_a}
\colon
\left( F_{(\blank)} A_{\idmorph_a}, {(\mu^F_{g,f,A_{\idmorph_a}})}_{g,f} \right)
\to
(A, \xi). \]
\begin{proposition}
The tuple $(\EMfree_a, \EMforget_a, \EMunit_a, \EMcounit_a)$
forms an adjunction.
\end{proposition}
\begin{proof}
  Follows by definition. For detail, see \cite{Street72}.
\end{proof}

For a morphism $f \colon a \to b$,
let us calculate the functor
$\EMforget_b \EM{F}f \EMfree_a \colon Fa \to Fb$.
For an object $X \in Fa$, we have
$\EMforget_b \EM{F}f \EMfree_a (X)
= \EMforget_b \EM{F}f \left( F_{(\blank)} X, {(\mu^F_{g,k,X})}_{g,k} \right) 
= \EMforget_b \left( F_{(\blank) \circ f} X, {(\mu^F_{g,k \circ f,X})}_{g,k} \right)
= F_{\idmorph_b \circ f} X = F_f X$.
For a morphism $x \colon X \to Y$, we have
$\EMforget_b \EM{F}f \EMfree_a (x)
= \EMforget_b \EM{F}f (F_{\blank} x)
= \EMforget_b (F_{\blank \circ f} x)
= F_{\idmorph_b \circ f} x = F_f x$.
Therefore, we have $F_f = \EMforget_b \EM{F}f \EMfree_a$.

\subsection{A Lax Functor Induced by Adjunctions and a Functor}
In \Cref{Eilenberg-Moore-construction-of-lax-functors},
a lax functor $F \colon \cats \to \cat$ gives adjunctions and a functor.
Conversely, we will show that adjunctions and a functor determine a lax functor.
The construction of lax functor
is a generalization of the construction \cite{FKM16,Mellies17} of lax $M$-action
from an adjunction and a strict $M$-action where $M$ is a monoid.

\begin{theorem}
  Given a functor $G \colon \cats \to \cat$,
  a map $F \colon \ob{\cats} \to \ob{\cat}$ and
  adjunctions
  $(J_a \colon Fa \to Ga, E_a \colon Ga \to Fa, \eta_a, \varepsilon_a)$
  for each $a \in \ob{\cats}$,
  $F$ is extended to a lax functor $F \colon \cats \to \cat$.
\end{theorem}
\begin{proof}
  For each morphism $f \colon a \to b$ of $\cats$,
  we define $Ff := E_b \circ Gf \circ J_a \colon Fa \to Fb$.
  The unit $\eta^F$ of $F$ is induced by the units of the adjunctions.
  For each $a \in \ob{\cats}$,
  we define $\eta^F_a := \eta_a \colon \idfunc_{Fa} \nattr F_{\idmorph_a}$.
  \[
  \begin{tikzcd}
    Fa
    \ar[r, "J_a"]
    \ar[d, dashed, "Ff"']
    & Ga
    \ar[d, "Gf"] \\
    Fb
    & Gb
    \ar[l, "E_b"]
  \end{tikzcd}
  =
  \begin{tikzpicture}[baseline=0.5cm, x=0.5cm, y=0.5cm]
  \coordinate (Jat) at (1,2);
  \coordinate (Gft) at (2,2);
  \coordinate (Ebt) at (3,2);
  \node[above] at (Jat) {$J_a$};
  \node[above] at (Gft) {$Gf$};
  \node[above] at (Ebt) {$E_b$};
  \coordinate (Jas) at (1,0);
  \coordinate (Gfs) at (2,0);
  \coordinate (Ebs) at (3,0);
  \node[below] at (Jas) {$J_a$};
  \node[below] at (Gfs) {$Gf$};
  \node[below] at (Ebs) {$E_b$};

  \fill[red!20] (0,0) -- (Jas) -- (Jat) -- (0,2) -- cycle;
  \fill[red!50] (Jas) -- (Gfs) -- (Gft) -- (Jat) -- cycle;
  \fill[blue!50] (Ebs) -- (Gfs) -- (Gft) -- (Ebt) -- cycle;
  \fill[blue!20] (4,0) -- (Ebs) -- (Ebt) -- (4,2) -- cycle;
  \draw (Jas) to (Jat);;
  \draw (Gfs) to (Gft);
  \draw (Ebs) to (Ebt);
\end{tikzpicture}
  , \quad
  \begin{tikzcd}
    Fa
    \ar[r, "J_a"]
    \ar[d, "\idfunc_{Fa}"', ""{name=A, below}]
    & Ga
    \ar[d, "{G \idmorph_a = \idfunc_{Ga}}", ""{name=B, above}] \\
    Fa
    & Ga
    \ar[l, "E_a"]
    \ar[from=A, to=B, shorten=3mm, Rightarrow, "\eta_a"]
  \end{tikzcd}
  =
  \begin{tikzpicture}[baseline=0.5cm, x=0.5cm, y=0.5cm]
  \coordinate (Ja) at (1,2);
  \coordinate (Ea) at (2,2);
  \coordinate (eta) at (1.5, 1);
  \coordinate (id) at (1.5, 0);
  \fill[red!20]
  (0,0) -- (3,0) -- (3,2) -- (Ea)
  to[out=-90, in=0] (eta)
  to[out=180, in=-90] (Ja) -- (0,2) -- cycle;
  \fill[red!50]
  (Ea) to[out=-90, in=0] (eta) to[out=180, in=-90] (Ja) -- cycle;
  \draw (Ja) to[out=-90, in=180] (eta) to[out=0, in=-90] (Ea);
  \node[above] at (Ja) {$J_a$};
  \node[above] at (Ea) {$E_a$};
  \node[below] at (id) {$\idmorph_{Fa}$};
  \node[below] at (eta) {$\eta_a$};
\end{tikzpicture}
  \]
  The multiplication $\mu^F$ of $F$ is induced by the counits of adjunctions.
  We define
  $\mu^F_{g,f} = J_a \varepsilon_b E_c \colon Fg \circ Ff \nattr F(g \circ f)$
  for each $f \colon a \to b$ and $g \colon b \to c$ of $\cats$.
  \[
  \begin{tikzcd}
    Ga
    \ar[r, "Gf", ""{name=At, above}]
    \ar[rrr, bend left=35, "G(g \circ f)", ""{name=Bt, above}]
    & Gb
    \ar[d, "E_b"']
    \ar[r, "\idfunc_{Gb}"{name=Bs}, ""{name=epst, above}]
    & Gb
    \ar[r, "Gg", ""{name=Ct, above}]
    & Gc
    \ar[d, "E_c"] \\
    Fa
    \ar[u, "J_a"]
    \ar[r, "Ff"', ""{name=As, below}]
    & Fb
    \ar[r, "\idfunc_{Fb}"', ""{name=epss, below}]
    & Fb
    \ar[u, "J_b"]
    \ar[r, "Fg"', ""{name=Cs, below}]
    & Fc
    \ar[from=As, to=At, shorten=4mm, equal]
    \ar[from=Bs, to=Bt, shorten=1mm, equal]
    \ar[from=Cs, to=Ct, shorten=4mm, equal]
    \ar[from=epss, to=epst, Rightarrow, shorten=4mm, "\varepsilon_b"]
  \end{tikzcd}
  =
  \begin{tikzpicture}[baseline=0.75cm, x=0.5cm, y=0.5cm]
  \coordinate (Jas) at (1,0);
  \coordinate (Gf) at (2,0);
  \coordinate (Ebs) at (3,0);
  \coordinate (Jbs) at (4,0);
  \coordinate (Gg) at (5,0);
  \coordinate (Ecs) at (6,0);
  \coordinate (Jat) at (2,3);
  \coordinate (Ggf) at (3.5,3);
  \coordinate (Ect) at (5,3);
  \coordinate (f_g) at (3.5, 2);
  \coordinate (eps) at (3.5,1);

  \fill[red!20]
  (0,0) -- (Jas) to[out=90, in=-90] (Jat) -- (0,3) --cycle;
  \fill[red!50]
  (Jas)
  to[out=90, in=-90] (Jat)
  -- (Ggf) -- (f_g)
  to[out=180, in=90] (Gf) --cycle;
  \fill[blue!50]
  (Gf)
  to[out=90, in=180] (f_g)
  to[out=0, in=90] (Gg) -- (Jbs)
  to[out=90, in=0] (eps)
  to[out=180, in=90] (Ebs) --cycle;
  \fill[blue!20]
  (Ebs) to[out=90, in=180] (eps) to[out=0, in=90] (Jbs) --cycle;
  \fill[green!50]
  (Ecs) to[out=90, in=-90] (Ect) -- (Ggf) -- (f_g)
  to[out=0, in=90] (Gg) --cycle;
  \fill[green!20]
  (7,0) -- (Ecs) to[out=90, in=-90] (Ect) -- (7,3) --cycle;
  
  \draw (Jas) to[out=90, in=-90] (Jat);
  \draw (Ecs) to[out=90, in=-90] (Ect);
  \draw (Ebs) to[out=90, in=180] (eps) to[out=0, in=90] (Jbs);
  \node[above] at (eps) {$\varepsilon_b$};
  \draw (Gf) to[out=90, in=180] (f_g);
  \draw (Gg) to[out=90, in=0] (f_g);
  \draw (f_g) -- (Ggf);
  \node[above] at (Jat) {$J_a$};
  \node[above] at (Ggf) {$G(g \circ f)$};
  \node[above] at (Ect) {$E_c$};
  \node[below] at (Jas) {$J_a$};
  \node[below] at (Gf) {$Gf$};
  \node[below] at (Ebs) {$E_b$};
  \node[below] at (Jbs) {$J_b$};
  \node[below] at (Gg) {$Gg$};
  \node[below] at (Ecs) {$E_c$};
\end{tikzpicture}
  \]
  
  We claim that $(F, \eta^F, \mu^F)$ is a lax functor.
  To show this claim,
  we have to show that the axioms of lax functors hold.
  The following equations of string diagrams imply the axioms hold.
  \[ \begin{tikzpicture}[baseline=1.25cm, x=0.5cm, y=0.5cm]
  \coordinate (Jas) at (1,0);
  \coordinate (Gf) at (2,0);
  \coordinate (Ebs) at (3,0);
  \coordinate (Jbs) at (4,0);
  \coordinate (Gg) at (5,0);
  \coordinate (Ecs) at (6,0);
  \coordinate (Jcs) at (7,0);
  \coordinate (Gh) at (8,0);
  \coordinate (Eds) at (9,0);
  
  \coordinate (Jat) at (3,5);
  \coordinate (Ghgf) at (5,5);
  \coordinate (Edt) at (7,5);
  \coordinate (f_g) at (3.5, 2);
  \coordinate (gf_h) at (5, 4);

  \coordinate (beps) at (3.5, 1);
  \coordinate (ceps) at (5, 3);

  \node[below] at (Jas) {$J_a$};
  \node[below] at (Gf) {$Gf$};
  \node[below] at (Ebs) {$E_b$};
  \node[below] at (Jbs) {$J_b$};
  \node[below] at (Gg) {$Gg$};
  \node[below] at (Ecs) {$E_c$};
  \node[below] at (Jcs) {$J_c$};
  \node[below] at (Gh) {$Gh$};
  \node[below] at (Eds) {$E_d$};
  
  \node[above] at (Jat) {$J_a$};
  \node[above] at (Ghgf) {$G(h \circ g \circ f)$};
  \node[above] at (Edt) {$E_d$};

  \fill[red!20]
  (0,5) -- (Jat)
  to[out=-90, in=90] (Jas) -- (0,0) --cycle;
  \fill[red!50]
  (Jat)
  to[out=-90, in=90] (Jas) -- (Gf)
  to[out=90, in=180] (f_g)
  to[out=90, in=180] (gf_h)
  to[out=90, in=-90] (Ghgf) -- cycle;
  \fill[blue!50]
  (Gf)
  to[out=90, in=180] (f_g)
  to[out=0, in=90] (Gg) -- (Jbs)
  to[out=90, in=0] (beps)
  to[out=180, in=90] (Ebs) --cycle;
  \fill[blue!20]
  (Jbs)
  to[out=90, in=0] (beps)
  to[out=180, in=90] (Ebs) --cycle;
  \fill[green!50]
  (Gg)
  to[out=90, in=0] (f_g)
  to[out=90, in=180] (gf_h)
  to[out=0, in=90] (Gh) -- (Jcs)
  to[out=90, in=0] (ceps)
  to[out=180, in=90] (Ecs)
  --cycle;
  \fill[green!20]
  (Jcs)
  to[out=90, in=0] (ceps)
  to[out=180, in=90] (Ecs)
  --cycle;
  \fill[yellow!50]
  (Gh)
  to[out=90, in=0] (gf_h) -- (Ghgf) -- (Edt)
  to[out=-90, in=90] (Eds) -- cycle;
  \fill[yellow!20]
  (10,5) -- (Edt)
  to[out=-90, in=90] (Eds) -- (10,0) --cycle;

  \draw (Ebs) to[out=90, in=180] (beps) to[out=0, in=90] (Jbs);
  \draw (Ecs) to[out=90, in=180] (ceps) to[out=0, in=90] (Jcs);
  \draw (Jas) to[out=90, in=-90] (Jat);
  \draw (Eds) to[out=90, in=-90] (Edt);
  \draw (Gf) to[out=90, in=180] (f_g);
  \draw (Gg) to[out=90, in=0] (f_g);
  \draw (f_g) to[out=90, in=180] (gf_h);
  \draw (Gh) to[out=90, in=0] (gf_h);
  \draw (gf_h) -- (Ghgf);

  \node[above] at (beps) {$\varepsilon_b$};
  \node[above] at (ceps) {$\varepsilon_c$};
\end{tikzpicture}
=
\begin{tikzpicture}[baseline=1.25cm, x=0.5cm, y=0.5cm]
  \coordinate (Jas) at (1,0);
  \coordinate (Gf) at (2,0);
  \coordinate (Ebs) at (3,0);
  \coordinate (Jbs) at (4,0);
  \coordinate (Gg) at (5,0);
  \coordinate (Ecs) at (6,0);
  \coordinate (Jcs) at (7,0);
  \coordinate (Gh) at (8,0);
  \coordinate (Eds) at (9,0);
  
  \coordinate (Jat) at (3,5);
  \coordinate (Ghgf) at (5,5);
  \coordinate (Edt) at (7,5);
  \coordinate (g_h) at (6.5, 2);
  \coordinate (f_hg) at (5, 4);

  \coordinate (beps) at (5, 3);
  \coordinate (ceps) at (6.5, 1);

  \node[below] at (Jas) {$J_a$};
  \node[below] at (Gf) {$Gf$};
  \node[below] at (Ebs) {$E_b$};
  \node[below] at (Jbs) {$J_b$};
  \node[below] at (Gg) {$Gg$};
  \node[below] at (Ecs) {$E_c$};
  \node[below] at (Jcs) {$J_c$};
  \node[below] at (Gh) {$Gh$};
  \node[below] at (Eds) {$E_d$};
  
  \node[above] at (Jat) {$J_a$};
  \node[above] at (Ghgf) {$G(h \circ g \circ f)$};
  \node[above] at (Edt) {$E_d$};

  \fill[red!20]
  (0,5) -- (Jat)
  to[out=-90, in=90] (Jas) -- (0,0) -- cycle;
  \fill[red!50]
  (Jat)
  to[out=-90, in=90] (Jas) -- (Gf)
  to[out=90, in=180] (f_hg) -- (Ghgf) -- cycle;
  \fill[blue!50]
  (Gf)
  to[out=90, in=180] (f_hg)
  to[out=0, in=90] (g_h)
  to[out=0, in=90] (Gg) -- (Jbs)
  to[out=90, in=0] (beps)
  to[out=180, in=90] (Ebs) --cycle;
  \fill[blue!20]
  (Jbs) to[out=90, in=0] (beps)
  to[out=180, in=90] (Ebs) --cycle;
  \fill[green!50]
  (Gg)
  to[out=90, in=180] (g_h)
  to[out=0, in=90] (Gh) -- (Jcs)
  to[out=90, in=0] (ceps)
  to[out=180, in=90] (Ecs) --cycle;
  \fill[green!20]
  (Jcs)
  to[out=90, in=0] (ceps)
  to[out=180, in=90] (Ecs) --cycle;
  \fill[yellow!50]
  (Gh)
  to[out=90, in=0] (g_h)
  to[out=90, in=0] (f_hg)
  -- (Ghgf) -- (Edt)
  to[out=-90, in=90] (Eds) -- cycle;
  \fill[yellow!20]
  (10,5) -- (Edt)
  to[out=-90, in=90] (Eds) -- (10,0) -- cycle;

  \draw (Ebs) to[out=90, in=180] (beps) to[out=0, in=90] (Jbs);
  \draw (Ecs) to[out=90, in=180] (ceps) to[out=0, in=90] (Jcs);
  \draw (Jas) to[out=90, in=-90] (Jat);
  \draw (Eds) to[out=90, in=-90] (Edt);
  \draw (Gf) to[out=90, in=180] (gf_h);
  \draw (Gg) to[out=90, in=180] (g_h);
  \draw (g_h) to[out=90, in=0] (f_hg);
  \draw (Gh) to[out=90, in=0] (g_h);
  \draw (f_hg) -- (Ghgf);

  \node[above] at (beps) {$\varepsilon_b$};
  \node[above] at (ceps) {$\varepsilon_c$};
\end{tikzpicture} \]
  \[ \begin{tikzpicture}[baseline=0.75cm, x=0.5cm, y=0.5cm]
  \coordinate (Jat) at (1,3);
  \coordinate (Gft) at (3,3);
  \coordinate (Ebt) at (5,3);
  \node[above] at (Jat) {$J_a$};
  \node[above] at (Gft) {$Gf$};
  \node[above] at (Ebt) {$E_b$};
  \coordinate (Jas) at (3,0);
  \coordinate (Gfs) at (4,0);
  \coordinate (Ebs) at (5,0);
  \node[below] at (Jas) {$J_a$};
  \node[below] at (Gfs) {$Gf$};
  \node[below] at (Ebs) {$E_b$};
  \coordinate (Jaeps) at (2.5, 1.5);
  \coordinate (Jaeta) at (1.5, 0.5);

  \fill[red!20]
  (0,0) -- (Jas)
  to[out=90, in=0] (Jaeps)
  to[out=180, in=0] (Jaeta)
  to[out=180, in=-90] (Jat)
  -- (0,3) -- cycle;
  \fill[red!50]
  (Jas)
  to[out=90, in=0] (Jaeps)
  to[out=180, in=0] (Jaeta)
  to[out=180, in=-90] (Jat) -- (Gft)
  to[out=-90, in=90] (Gfs) --cycle;
  \fill[blue!50]
  (Gfs) to[out=90, in=-90] (Gft) -- (Ebt) -- (Ebs) --cycle;
  \fill[blue!20]
  (6,0) -- (Ebs) -- (Ebt) -- (6,3) -- cycle;

  \node[below] at (Jaeta) {$\eta_a$};
  \node[above] at (Jaeps) {$\varepsilon_a$};
  \draw
  (Jas) to[out=90, in=0]
  (Jaeps) to[out=180, in=0]
  (Jaeta) to[out=180, in=-90]
  (Jat);
  \draw (Gfs) to[out=90, in=-90] (Gft);
  \draw (Ebs) to (Ebt);
\end{tikzpicture}
=
\begin{tikzpicture}[baseline=0.75cm, x=0.5cm, y=0.5cm]
  \coordinate (Jat) at (1,3);
  \coordinate (Gft) at (2,3);
  \coordinate (Ebt) at (3,3);
  \node[above] at (Jat) {$J_a$};
  \node[above] at (Gft) {$Gf$};
  \node[above] at (Ebt) {$E_b$};
  \coordinate (Jas) at (1,0);
  \coordinate (Gfs) at (2,0);
  \coordinate (Ebs) at (3,0);
  \node[below] at (Jas) {$J_a$};
  \node[below] at (Gfs) {$Gf$};
  \node[below] at (Ebs) {$E_b$};
  
  \fill[red!20] (0,0) -- (Jas) -- (Jat) -- (0,3) -- cycle;
  \fill[red!50] (Jas) -- (Jat) -- (Gft) -- (Gfs) -- cycle;
  \fill[blue!50] (Gft) -- (Gfs) -- (Ebs) -- (Ebt) -- cycle;
  \fill[blue!20] (4,0) -- (Ebs) -- (Ebt) -- (4,3) -- cycle;
  \draw (Jas) to (Jat);;
  \draw (Gfs) to (Gft);
  \draw (Ebs) to (Ebt);
\end{tikzpicture}
=
\begin{tikzpicture}[baseline=0.75cm, x=0.5cm, y=0.5cm]
  \coordinate (Jat) at (1,3);
  \coordinate (Gft) at (3,3);
  \coordinate (Ebt) at (5,3);
  \node[above] at (Jat) {$J_a$};
  \node[above] at (Gft) {$Gf$};
  \node[above] at (Ebt) {$E_b$};
  \coordinate (Jas) at (1,0);
  \coordinate (Gfs) at (2,0);
  \coordinate (Ebs) at (3,0);
  \node[below] at (Jas) {$J_a$};
  \node[below] at (Gfs) {$Gf$};
  \node[below] at (Ebs) {$E_b$};
  \coordinate (Ebeps) at (3.5, 1.5);
  \coordinate (Ebeta) at (4.5, 0.5);

  \fill[red!20]
  (0,0) -- (Jas) -- (Jat) -- (0,3) --cycle;
  \fill[red!50]
  (Gfs)
  to[out=90, in=-90] (Gft) -- (Jat) -- (Jas) -- cycle;
  \fill[blue!50]
  (Gfs) to[out=90, in=-90] (Gft) -- (Ebt)
  to[out=-90, in=0] (Ebeta)
  to[out=180, in=0] (Ebeps)
  to[out=180, in=90] (Ebs) -- cycle;
  \fill[blue!20]
  (6,3) -- (Ebt)
  to[out=-90, in=0] (Ebeta)
  to[out=180, in=0] (Ebeps)
  to[out=180, in=90] (Ebs) -- (6,0) --cycle;

  \node[below] at (Ebeta) {$\eta_b$};
  \node[above] at (Ebeps) {$\varepsilon_b$};
  \draw (Jas) to (Jat);
  \draw (Gfs) to[out=90, in=-90] (Gft);
  \draw
  (Ebs) to[out=90, in=180]
  (Ebeps) to[out=0, in=180]
  (Ebeta) to[out=0, in=-90]
  (Ebt);
\end{tikzpicture} \]
\end{proof}

\begin{corollary}\label{cor:lax-functor-from-adjunctions-and-a-functor}
  Given a functor $G \colon \cats^{\opposite} \to \cat$,
  a category $\catc$ and
  adjunctions
  $(J_a \colon \catc \to Ga, E_a \colon Ga \to \catc, \eta_a, \varepsilon_a)$
  for each $a \in \ob{\cats}$,
  there exists an $\cats$-graded monad $T \colon \cats^{\opposite} \to \Endo{\catc}$
  such that $T_f = E_b Gf J_a$
  for each $f\colon b \to a$ in $\cats$.
\end{corollary}

\section{Category-Graded Algebraic Theories}\label{category-graded-algebraic-theory}
We explain category-graded extensions of algebraic theories.
In this section, we fix a small category $\cats$ and
a category $\catc$ with countable products.

\subsection{Category-Graded Terms}
In a category-graded algebraic theory,
each term is graded by a morphism in a grading category $\cats$.
This is analogous to parameterised and (monoid-)graded algebraic theories
\cite{Atkey09,Kura20}.

\begin{definition}[Signature]
  An \textit{($\cats$-graded) signature} $\Sigma$ is a set of symbols.
  For each $\sigma \in \Sigma$,
  countable or finite sets $P$ and $A$,
  and a morphism $f$ in $\cats$ are assigned.
  $\sigma \in \Sigma$ is called an \textit{operation}.
  $P$, $A$ and $f$ are called
  a \textit{parameter}, \textit{arity} and \textit{grade}
  of $\sigma$, respectively.
  We write
  $ \sigma \colon \oprtyp{P}{A}{f}$
  for an operation $\sigma$ whose
  parameter, arity and grade are
  $P$, $A$, $f$, respectively.
\end{definition}

\begin{definition}[$\Sigma$-term]
  Let $X$ be a set.
  The set of \textit{$\Sigma$-terms} $\Term{\Sigma}{f}{X}$
  for each $f \colon b \to a$ in $\cats$
  is defined recursively as follows.
  \[
  \infer{
    \pure(a, x) \in \Term{\Sigma}{\idmorph_a}{X}
  }{
    a \in \ob{\cats}
    &
    x \in X
  }
  \]
  \[
  \infer{
    \doop(\sigma, p, {\{ t_i \}}_{i \in A} ) \in \Term{\Sigma}{f ; g}{X}
  }{
    p \in P
    &
    \sigma \colon \oprtyp{P}{A}{f \colon c \to b}
    &
    {\{ t_i \}}_{i \in A} \subset \Term{\Sigma}{g \colon b \to a}{X}
  }
  \]
\end{definition}
We sometimes write
$\doop(\sigma, p, \lambda i. t_i )$
instead of
$\doop(\sigma, p, {\{ t_i \}}_{i \in A} )$.
Intuitively, $\doop(\sigma, p, \lambda i . t_i)$ is the term
that performs the operation $\sigma$ with parameter $p$,
binds the result to $i$,
and invokes the continuation $t_i$.
Note that, when $A$ is the arity of $\sigma$,
the term $\doop(\sigma, p, \lambda i . t_i)$ is a term that takes $A$-many terms $\{t_i\}_{i \in A}$ as arguments.

\begin{definition}[$\Sigma$-model]
  Let $\Sigma$ be a signature.
  A \textit{$\Sigma$-model}
  $I = (I, \interp{I}{\blank})$ at $a \in \cats$
  is a pair of a map
  $I \colon \prod_{b \in \cats} \cats(b, a) \to \catc$
  and an interpretation
  $\interp{I}{\sigma}_{k} \colon P \times {I(k)}^A \to I(f ; k)$
  for each operation $\sigma \colon \oprtyp{P}{A}{f \colon c \to b} \in \Sigma$
  and morphism $k \colon b \to a$ in $\cats$.

  A \textit{homomorphism} $\alpha \colon I \to J$
  between two $\Sigma$-models $I$ and $J$ at $a$ is
  a family of morphisms ${\{ \alpha_k \colon I(k) \to J(k) \}}_{k \colon b \to a}$
  such that
  for every operation $\sigma \colon \oprtyp{P}{A}{f \colon b \to c}$
  and morphism $k \colon c \to a$,
  the following diagram commutes:
  \[ \begin{tikzcd}
    P \times {I(k)}^A
    \ar[r, "P \times \alpha_k^A"]
    \ar[d, "\interp{I}{\sigma}_k"']
    & P \times {J(k)}^A
    \ar[d, "\interp{J}{\sigma}_k"] \\
    I(f ; k)
    \ar[r, "\alpha_{f ; k}"']
    & J(f ; k)
  \end{tikzcd} \]
\end{definition}
The map $I : \prod_b \cats(b,a) \to \catc$ assigns a ``carrier set'' $I(k)$ to each $k : b \to a$.
Given a $\Sigma$-model $I$,
we can interpret $\Sigma$-terms by extending the interpretation of $I$.

\begin{definition}[Interpretation of $\Sigma$-terms]
  Let $\Sigma$ be a signature and $I$ be a $\Sigma$-model at $a \in \ob{\cats}$.
  For each $\Sigma$-term $t \in \Term{\Sigma}{f \colon b \to a}{X}$,
  the element $\interp{I}{t}$, called its \textit{interpretation}, of
  the set $\prod_{a' \in \cats, k \in \cats(a,a')} \catc({I(k)}^X, I(f;k))$
  is defined recursively as
  \[ \begin{split}
    \interp{I}{\pure(a, x)} & \defeq {\{ \pi_x \colon {I(k)}^X \to I(k) \}}_{k \colon a \to a'}, \\
    \interp{I}{\doop(\sigma, p, \lambda i. t_i)}
    & \defeq {\{ (p \times {\langle \interp{I}{t_i}_k \rangle}_{i \in A});
    \interp{I}{\sigma}_{g; k}
    \colon {I(k)}^X \to I(f; g; k) \}}_{k \colon a \to a'}
  \end{split} \]
  where
  $(\sigma \colon \oprtyp{P}{A}{f : c \to b}) \in \Sigma$
  and $t_i \in \Term{\Sigma}{g \colon b \to a}{X}$.
\end{definition}

Intuitively,
grading morphisms are sequences of sorts of effects
that will be invoked by executing terms.

\begin{example}
  Category graded algebraic theories are useful
  to deal with ``order-sensitive'' operations.
  To illustrate this, we provide an example that contains operations
  for mutable state and sending and receiving data.
  In this example,
  grading morphisms represent orders of sending and receiving effect and types of data,
  analogously to session types.
  Let $\cats$ be a category
  whose objects are base types and
  morphisms are generated by
  $\alpha \xto{\recv^{\alpha}_{\beta}} \beta$,
  $\alpha \xto{\send_{\alpha}} \alpha$,
  $\alpha \xto{\tau^{\alpha}_{\beta}} \beta$,
  $\tau^{\beta}_{\gamma} \circ \tau^{\alpha}_{\beta} = \tau^{\alpha}_{\gamma}$,
  and $\tau^{\alpha}_{\alpha} = \idmorph_\alpha$.
  Let
  $\recvint_{\alpha} \colon \oprtyp{\unit}{\unit}{\alpha \xto{\recv} \num}$,
  $\sendint \colon \oprtyp{\unit}{\unit}{\num \xto{\send} \num}$,
  $\lookupint \colon \oprtyp{\unit}{\num}{\num \xto{\idmorph_{\num}} \num}$,
  $\updateint_{\alpha} \colon \oprtyp{\num}{\unit}{\alpha \xto{\tau} \num}$,
  and
  $\Sigma$ be
  ${\{ \recvint_{\alpha}, \updateint_{\alpha} \}}_{\alpha} \cup \{ \sendint, \lookupint \}$.
  We have the following $\Sigma$-terms:
  \[ \begin{split}
    t \defeq &
    \doop(\updateint, 2, \lambda \_ .
    \doop(\sendint, \unitval, \lambda \_ .
    \doop(\recvint_{\num}, \unitval, \lambda \_ . \\
    & \doop(\lookupint, \unitval, \lambda n .
    \pure(\num, n)))))
    \in \Term{\Sigma}{ \tau ; \send_{\num} ; \recv^{\num}_{\num} }{\num},
  \end{split} \]
  \[ \begin{split}
    s \defeq &
    \doop(\recvint_{\unit}, \unitval, \lambda \_ .
    \doop(\lookupint, \unitval, \lambda n .
    \doop(\updateint, n+1, \lambda \_ . \\
    & \doop(\sendint_{\num}, \unitval, \lambda \_ .
    \pure(\num, \unitval)))))
    \in \Term{\Sigma}{\recv^{\unit}_{\num} ; \send_{\num}}{\unit}.
  \end{split} \]
  The term $t$ is graded by
  the morphism
  $\unit \xto{\tau} \num \xto{\send_{\num}} \num \xto{\recv^{\num}_{\num}} \num$
  in $\cats$.
  This means that the term $t$ executes internal effect $\tau$,
  sends data of type $\num$,
  and then receives data of type $\num$.
  The term $s$ is graded by
  the morphism
  $\unit \xto{\recv^{\unit}_{\num}} \num \xto{\send_{\num}} \num$.
  This means that the term $s$
  receives data of type $\num$ and then
  sends data of type $\num$.
  We can know from grading morphisms of $t$ and $s$ that they can interact with each other.
\end{example}

\subsection{Equations}
Next, we introduce equations to represent
the equational theory on terms.
The equation is defined as a pair of terms as in the non-graded case.
However, the pairs of terms must have the same grading morphism.

\begin{definition}[Equations and category-graded algebraic theory]
  A graded family of \textit{equations} for $\Sigma$ is a family of sets
  $E = (E_f)_{f}$
  where $E_f$ is a set of pairs of terms in $\Term{\Sigma}{f}{X}$.
  We write $t = s$ for a pair $(t, s) \in E_f$.
%
  An \textit{$\cats$-graded algebraic theory} is a pair
  $\gat{T} = (\Sigma, E)$
  of $\cats$-graded signature $\Sigma$
  and equations $E$ for $\Sigma$.
\end{definition}

\begin{definition}[Model for category-graded algebraic theory]
  Let $\gat{T} = (\Sigma, E)$ be an $\cats$-graded algebraic theory
  and $a$ be an object of $\cats$.
  A model for $\gat{T}$ at $a$ is
  a $\Sigma$-model $I$ at $a$ that satisfies
  $\interp{I}{t} = \interp{I}{s}$
  for each morphism $f \colon c \to b$ in $\cats$ and equation $t = s \in E_f$.
  We denote the category of models for $\gat{T}$ at $a$ by
  $\GATMod{\gat{T}}{a}{\catc}$.
\end{definition}

\subsection{Free Models and Adjunctions}
We explain free models of a category-graded algebraic theory
and its universal property.

\begin{definition}
  Let $\gat{T} = (\Sigma, E)$ be an $\cats$-graded algebraic theory.
  We define a functor
  $\free{\gat{T}}{a} \colon \sets \to \GATMod{\gat{T}}{a}{\sets}$
  by
  $(\free{\gat{T}}{a} X) (k) = \Term{\Sigma}{k}{X}/\sim$
  for $k \colon b \to a$ in $\cats$
  and
  $\interp{\free{\gat{T}}{a}X}{\sigma}_k (p, {\{ [t_i] \}}_{i \in A})
  = [\doop(\sigma, p, \lambda i. t_i)]$
  for each $X \in \sets$,
  $k \colon b \to a$ in $\cats$
  and $\sigma \colon \oprtyp{P}{A}{f}$
  where
  $\sim$ is the equivalence relation induced by the equations $E$ and
  $[t]$ is the equivalence class of $t$.
  We also define a map $\eta_X \colon X \to (\free{\gat{T}}{a}{X})(\idmorph_a)$
  by $\eta_X(x) = [\pure(a,x)] \in \Term{\Sigma}{\idmorph_a}{X} / \sim$.
\end{definition}

We can show that the model $\free{\gat{T}}{a}{X}$ with $\eta_X$
has the universal property of a free model.

\begin{proposition}
  Let $\gat{T} = (\Sigma, E)$ be an $\cats$-graded algebraic theory.
  Given a model $A$ in $\sets$ for $\gat{T}$ at $a$
  and a map $\phi \colon X \to A \idmorph_a$,
  there exists a unique homomorphism
  $\freeex{\phi} \colon \free{\gat{T}}{a}X \to A$
  such that
  $\freeex{\phi}_{\idmorph_a} \circ \eta_X = \phi$.
  \[
  \begin{tikzcd}
    X
    \ar[r, "\eta_X"]
    \ar[rd, "\phi"']
    & (\free{\gat{T}}{a}X)(\idmorph_a)
    \ar[d, "\freeex{\phi}_{\idmorph_a}"] \\
    & A \idmorph_a
  \end{tikzcd}
  \quad
  \begin{tikzcd}
    \free{\gat{T}}{a}X
    \ar[d, dashed, "\freeex{\phi}"] \\
    A
  \end{tikzcd}
  \]
\end{proposition}
\begin{proof}
  For each $f \colon b \to a$,
  we define a map
  $ \hat{\phi}_f \colon \Term{\Sigma}{f}{X} \to Af$
  from the set of $\Sigma$-terms to $Af$
  recursively by:
  $\hat{\phi}_{\idmorph_a}(\pure(a, x)) = \phi(x)$ and
  $\hat{\phi}_f( \doop(\sigma, p, \lambda i. t_i) ) = \interp{A}{\sigma}(p, \lambda i. \hat{\phi}(t_i))$.
  Since $A$ is a model for $\gat{T}$, all equations in $E$ holds in $A$.
  Therefore, the map $\freeex{\phi}([t]) \defeq [\hat{\phi}(t)]$ is well-defined.
  We can show $\freeex{\phi}_{\idmorph_a} \circ \eta_X = \phi$ by definition.
\end{proof}

The forgetful functor for models
$\forget{\gat{T}}{a} \colon \GATMod{\gat{T}}{a}{\sets} \to \sets$ is
the evaluation at $\idmorph_a$, that is
$\forget{\gat{T}}{a} (A) = A_{\idmorph_a}$.
By the universality of the free construction,
we have isomorphisms
$
\GATMod{\gat{T}}{a}{\sets} (\free{\gat{T}}{a}X, A)
\cong
\catc(X, \forget{\gat{T}}{a}A)
$
for each $a \in \ob{\cats}$.
Thus, we have adjunctions $\free{\gat{T}}{a} \dashv \forget{\gat{T}}{a}$.
\[
\begin{tikzcd}[column sep=8em]
  \catc
  \ar[r, bend left=20, "\free{\gat{T}}{a}", ""{name=A, below}]
  & \GATMod{\gat{T}}{a}{\sets}
  \ar[l, bend left=20, "\forget{\gat{T}}{a}", ""{name=B,above}]
  \arrow[phantom, from=A, to=B, "\dashv" rotate=-90]
\end{tikzcd}.
\]

We can define
$\GATMod{\gat{T}}{k}{\sets}
\colon \GATMod{\gat{T}}{a}{\sets} \to \GATMod{\gat{T}}{b}{\sets}$
for a morphism $k \colon b \to a$ in $\cats$
to make $\GATMod{\gat{T}}{\blank}{\sets}$
a functor $\cats^{\op} \to \cat$.
For a model $A$ in $\GATMod{\gat{T}}{a}{\sets}$,
a map
$\GATMod{\gat{T}}{k}{\sets}A \colon \prod_c \cats(c,b) \to \sets$
is defined as
$(\GATMod{\gat{T}}{k}{\sets}A) (f) = A(f ; k)$.
Interpretation of operations
$\interp{ \GATMod{\gat{T}}{k}{\sets}A }{\blank}$
is defined as
$\interp{ \GATMod{\gat{T}}{k}{\sets} A}{\sigma}_f
= \interp{A}{\sigma}_{f ; k}$.
It is easy to check that
$\GATMod{\gat{T}}{k}{\sets}$ is a functor
from $\GATMod{\gat{T}}{a}{\sets}$ to $\GATMod{\gat{T}}{b}{\sets}$.

More generally, we can define a category of models $\GATMod{\gat{T}}{a}{\catc}$ for a category $\catc$ in the same way as $\catc = \sets$,
and we can apply the same argument as above if the left adjoint of forgetful functor exists.

Applying Corollary \ref{cor:lax-functor-from-adjunctions-and-a-functor},
we obtain an $\cats$-graded monad $\cats^{\op} \to \Endo{\catc}$.
We denote the category-graded monad by $T^{\gat{T}}$.
The unit and multiplication of $T^{\gat{T}}$ are depicted as follows.
\[
\eta^{T^{\gat{T}}}_a = \begin{tikzpicture}[baseline=0.5cm, x=0.5cm, y=0.5cm]
  \coordinate (Fa) at (1,2);
  \coordinate (Ua) at (2,2);
  \coordinate (eta) at (1.5, 1);
  \coordinate (id) at (1.5, 0);

  \fill[black!20]
  (0,0) -- (3,0) -- (3,2) -- (Ua)
  to[out=-90, in=0] (eta)
  to[out=180, in=-90] (Fa) -- (0,2) -- cycle;
  \fill[red!50]
  (Ua) to[out=-90, in=0] (eta) to[out=180, in=-90] (Fa) -- cycle;

  \draw (Fa) to[out=-90, in=180] (eta) to[out=0, in=-90] (Ua);
  \node[above] at (Fa) {$\free{\gat{T}}{a}$};
  \node[above] at (Ua) {$\forget{\gat{T}}{a}$};
  \node[below] at (id) {$\idfunc_{\catc}$};
\end{tikzpicture},
\quad
\mu^{T^{\gat{T}}}_{g,f} = \begin{tikzpicture}[baseline=0.75cm, x=0.5cm, y=0.5cm]
  \coordinate (Jas) at (1,0);
  \coordinate (Mf) at (2,0);
  \coordinate (Ebs) at (3,0);
  \coordinate (Jbs) at (4,0);
  \coordinate (Mg) at (5,0);
  \coordinate (Ecs) at (6,0);
  \coordinate (Jat) at (2,3);
  \coordinate (Mgf) at (3.5,3);
  \coordinate (Ect) at (5,3);
  \coordinate (f_g) at (3.5, 2);
  \coordinate (eps) at (3.5,1);

  \fill[black!20]
  (0,0) -- (Jas) to[out=90, in=-90] (Jat) -- (0,3) --cycle;
  \fill[red!50]
  (Jas)
  to[out=90, in=-90] (Jat)
  -- (Mgf) -- (f_g)
  to[out=180, in=90] (Mf) --cycle;
  \fill[blue!50]
  (Mf)
  to[out=90, in=180] (f_g)
  to[out=0, in=90] (Mg) -- (Jbs)
  to[out=90, in=0] (eps)
  to[out=180, in=90] (Ebs) --cycle;
  \fill[black!20]
  (Ebs) to[out=90, in=180] (eps) to[out=0, in=90] (Jbs) --cycle;
  \fill[green!50]
  (Ecs) to[out=90, in=-90] (Ect) -- (Mgf) -- (f_g)
  to[out=0, in=90] (Mg) --cycle;
  \fill[black!20]
  (7,0) -- (Ecs) to[out=90, in=-90] (Ect) -- (7,3) --cycle;
  
  \draw (Jas) to[out=90, in=-90] (Jat);
  \draw (Ecs) to[out=90, in=-90] (Ect);
  \draw (Ebs) to[out=90, in=180] (eps) to[out=0, in=90] (Jbs);
  \draw (Mf) to[out=90, in=180] (f_g);
  \draw (Mg) to[out=90, in=0] (f_g);
  \draw (f_g) -- (Mgf);

  \node[above] at (Jat) {$\free{\gat{T}}{a}$};
  \node[above=1em] at (Mgf) {$\GATMod{\gat{T}}{f;g}{\catc}$};
  \node[above] at (Ect) {$\forget{\gat{T}}{c}$};
  \node[below] at (Jas) {$\free{\gat{T}}{a}$};
  \node[below=1em] at (Mf) {$\GATMod{\gat{T}}{f}{\catc}$};
  \node[below] at (Ebs) {$\forget{\gat{T}}{b}$};
  \node[below] at (Jbs) {$\free{\gat{T}}{b}$};
  \node[below=1em] at (Mg) {$\GATMod{\gat{T}}{g}{\catc}$};
  \node[below] at (Ecs) {$\forget{\gat{T}}{c}$};
\end{tikzpicture}.
\]

\section{A Category-Graded Effect System}\label{category-graded-effect-system}
In this section,
we introduce an effect system with category-graded operations
based on category-graded algebraic theories.
We call the effect system \cateff{}.
We also construct handlers of category-graded algebraic operations.
We fix small categories $\cats$ and $\cats'$,
and $\cats$, $\cats'$-signatures $\Sigma$, $\Sigma'$.

\subsection{Language}
\textbf{Syntax}.
Our effect system \cateff{} is based on fine-grained call-by-value calculus \cite{LVT03}.
The syntax is divided into two parts, values and computations.
\[
\text{\textbf{Values}}\quad
V, W ::=
x \mid \unitval
\mid \inl{V} \mid \inr{V}
\mid \pair{V}{W}
\mid \lambda^f x : A . M
\]
\[
\begin{split}
  \text{\textbf{Computations}} \quad
  M, N ::=
  & \ret_a V
  \mid \letin{x}{M}{N}
  \mid V W
  \mid \opr(V) \\
  & \mid \proj{V}{x}{y}{M}
  \mid \match{V}{x}{M_1}{y}{M_2}
\end{split}
\]
where $a$ is an object in $\cats$, and $f$ is a morphism in $\cats$.
Values are usual except for lambda abstraction.
The lambda abstraction $\lambda^f x : A. M$ means that
this function has effects represented by the morphism $f$.
We sometimes write $V_{\Sigma}$ and $M_{\Sigma}$ to specify its signature.

\textbf{Types}.
Types are defined by the following BNF:
\[
  P, Q ::= \unit
  \mid P \times Q
  \mid P + Q, \quad
  A, B ::= \unit
  \mid A \times B
  \mid A + B
  \mid A \to B; a \xto{f} b.
\]
where $f \colon a \to b$ is a morphism in $\cats$.
The key idea is that
a function type $A \to B ; a \xto{f} b$ indicates that
a function of this type consumes an argument of type $A$ and returns a result of type $B$
with effects represented by $f$.
We call $P$ a \textit{primitive type}.
We assume that
parameters and arities of operations in the signatures $\Sigma$ and $\Sigma'$
are primitive types.

\textbf{Typing rules}.
There are
value judgements of the form $\Gamma \vdash V : A$
and computation judgements of the form $\Gamma \vdash^{\Sigma}_{f} M : A$
where $\Gamma$ is a list of distinct variables with types
and $f$ is a morphism in $\cats$.
The judgement $\Gamma \vdash^{\Sigma}_f M : A$ means that
the computation $M$ returns a result of type $A$ under type environment $\Gamma$
and causes effects represented by $f$.
We omit $\Sigma$ in the judgement $\Gamma \vdash^{\Sigma}_f M : A$
if it is clear from the context.
The typing rules for terms in \cateff{} are presented in \Cref{typing-rules}.

\begin{figure}[htp]
  \centering
  \textbf{Values}
  \[
  \infer[\textsc{Tv-Unit}]{
    \Gamma \vdash \unitval : \unit
  }{
  }
  \quad
  \infer[\textsc{Tv-Var}]{
    \Gamma \vdash x : A
  }{
    x:A \in \Gamma
  }
  \quad
  \infer[\textsc{Tv-Abs}]{
    \Gamma \vdash \lambda^f x. M : A \to B ; f
  }{
    \Gamma, x : A \vdash_{f} M : B
  }
  \]
  \[
    \infer[\textsc{Tv-Pair}]{
    \Gamma \vdash \pair{V}{W} : A \times B
  }{
    \Gamma \vdash V : A
    &
    \Gamma \vdash W : B
  }
  \]
  \[
  \infer[\textsc{Tv-InjL}]{
    \Gamma \vdash \inl V : A + B
  }{
    \Gamma \vdash V : A
  }
  \quad
  \infer[\textsc{Tv-InjR}]{
    \Gamma \vdash \inr V : A + B
  }{
    \Gamma \vdash V : B
  }
  \]

  \textbf{Computations}
  \[
  \infer[\textsc{Tc-Val}]{
    \Gamma \vdash_{\idmorph_a} \ret_{a} V : A
  }{
    \Gamma \vdash V : A
  }
  \quad
  \infer[\textsc{Tc-Op}]{
    \Gamma \vdash_f \op(V) : Q
  }{
    \Gamma \vdash V : P
    &
    \op \colon \oprtyp{P}{Q}{a \xto{f} b \in \Sigma}
  }
  \]
  \[
  \infer[\textsc{Tc-Let}]{
    \Gamma \vdash_{f ; g} \letin{x}{M}{N} : B
  }{
    \Gamma \vdash_{f : a \to b} M : A
    &
    \Gamma, x : A \vdash_{g : b \to c} N : B
  }
  \quad
  \infer[\textsc{Tc-App}]{
    \Gamma \vdash_f VW : B
  }{
    \Gamma \vdash V : A \to B ; f
    &
    \Gamma \vdash W : A
  }
  \]
  \[
  \infer[\textsc{Tc-Proj}]{
    \Gamma \vdash_f
    \proj{V}{x}{y}{M} : B
  }{
    \Gamma \vdash
    V : A_1 \times A_2
    &
    \Gamma, x : A_1, y : A_2 \vdash_f
    M : B
  }
  \]
  \[
  \infer[\textsc{Tc-Match}]{
    \Gamma \vdash_f
    \match{V}{x}{M_1}{y}{M_2} : B
  }{
    \Gamma \vdash V : A_1 + A_2
    &
    \Gamma, x : A_1 \vdash_f M_1 : B
    &
    \Gamma, y : A_2 \vdash_f M_2 : B
  }
  \]
  \caption{Typing rules.}
  \label{typing-rules}
\end{figure}

\subsection{Handlers}
Handlers for ordinary algebraic theories \cite{PlPr09}
are homomorphisms from a free model for a theory to another one.
We can also construct handlers for category-graded algebraic theories
in a similar way to the ordinary handlers.
Let $G \colon \cats \to \cats'$ be a functor
and $\gat{T} = (\Sigma, E)$, $\gat{T}' = (\Sigma', E')$
be $\cats$-, and $\cats'$-graded algebraic theories, respectively.
For an object $a$ of $\cats$ and sets $X$ and $Y$,
a handler from $\free{\gat{T}}{a}X$ to $(\free{\gat{T}'}{Ga}Y)(G \blank)$ is obtained by
the universality of the free model $\free{\gat{T}}{a}X$.
To obtain the handler,
$(\free{\gat{T}'}{Ga}Y)(G \blank)$ must be a model for $\gat{T}$ at $a$.
So we need the following data:
\begin{itemize}
\item a map $\phi \colon X \to (\free{\gat{T}'}{Ga}Y)(G \idmorph_a)$, and
\item interpretations
  $\interp{(\free{\gat{T}'}{Ga}Y)G}{\sigma}_k \colon
  P \times {(\free{\gat{T}'}{Ga}Y)(Gk)}^Q \to (\free{\gat{T}'}{Ga}Y)(G(f ; k))$
  of operations in $\Sigma$
  for every $\sigma \colon \oprtyp{P}{Q}{c \xto{f} b} \in \Sigma$ and $k \colon b \to a$.
\end{itemize}
satisfying all equations in $E$, that is
the interpretations of terms in $\Term{\Sigma}{f}{X}$ induced by
$\interp{(\free{\gat{T}'}{Ga}Y)G}{\sigma}$
satisfies
$\interp{(\free{\gat{T}'}{Ga}Y)G}{t}=
\interp{(\free{\gat{T}'}{Ga}Y)G}{s}$
for every $t = s \in E$.

Together with the above data,
$(\free{\gat{T}'}{Ga}Y)(G \blank)$ becomes a model for $\gat{T}$ at $a$,
and there is a homomorphism
$\freeex{\phi} \colon \free{\gat{T}}{a}{X} \to (\free{\gat{T}'}{Ga}Y)G$
such that
$\freeex{\phi}_{\idmorph_a} \circ \eta_X = \phi$
by the universal property of free model.
We can calculate the homomorphism
$\freeex{\phi} = {\{ \freeex{\phi}_l \}}_{l \colon b \to a}$ as
$\freeex{\phi}_{\idmorph_a}([\pure(a, v)]) = \phi(v)$ and
$\freeex{\phi}_{l \colon b \to a}([\doop(\sigma, p, \lambda i. t_i)])
= \interp{(\free{\gat{T}'}{Ga}Y)G}{\sigma}_{k}(p, \lambda i. \freeex{\phi}_{k}([t_i]))$
where $t_i \in (\free{\gat{T}'}{Ga} Y)(Gk)$.

We add syntax and typing rules for handlers as follows.
Note that the following constructions of handlers
don't care about equations of algebraic theories.
Programmers must ensure on their responsibilities that
handlers they are constructing respect proper equations of effects.

\textbf{Additional syntax}.
To construct handlers for \cateff{},
we need data that corresponds to
$\phi \colon X \to  (\free{\gat{T}'}{Ga}Y)(G \idmorph_a)$ and
$\interp{(\free{\gat{T}'}{Ga}Y)G}{\sigma}_k
\colon P \times {(\free{\gat{T}'}{Ga}Y)(Gk)}^Q \to (\free{\gat{T}'}{Ga}Y)(G(f ; k))$
as argued above.
Thus, we extend the syntax of \cateff{} as follows:
\[
\begin{array}{lrcl}
  \text{\textbf{Computations}} & M_{\Sigma'} & ::= &
  \dots \mid \handle{M_{\Sigma}}{H_{\Sigma \Rightarrow \Sigma'}}  \\
  \text{\textbf{Handlers}} &
  H_{\Sigma \Rightarrow \Sigma'} & ::= &
  \{ \ret_b x \mapsto M_{\Sigma'} \} \\
  & & &
  \cup {\{ \op(p), r \mapsto_k {(M_{\Sigma'})}^k_{\op} \}}^{k \colon c \to b}_{
    \op \colon \oprtyp{P}{Q}{d \xto{g} c} \in \Sigma}
\end{array}
\]

\textbf{Additional typing rules}.
We add a new judgement $\Gamma \vdash^G_b H : R \handlerto R'$
where $R$ and $R'$ are primitive types.
This means that
the handler $H$ handles operations in computation of type $R$ and
then produces a computation of type $R'$.
Let $\Delta$ be an environment $x_1 : P_1, \dots, x_n : P_n$.
\[
\infer[\textsc{Tc-Handle}]{
  \Gamma \vdash^{\Sigma'}_{G(f)} \handle{M}{H} : R'
}{
  \Delta \vdash^{\Sigma}_{f \colon a \to b} M : R
  &
  \Gamma \vdash^G_b H : R \handlerto R'
  &
  \Delta \subseteq \Gamma
}
\]
\[
\infer[\textsc{Th-Handler}]{
  \Gamma \vdash^G_b
  \{ \ret_b x \mapsto M \} \cup {\{ \op(p), r \mapsto_k M^k_{\op} \}}^k_{\op \in \Sigma}
  : R \handlerto R'
}{
  \begin{array}{c}
    \Gamma, x : R \vdash^{\Sigma'}_{\idmorph_{Gb}} M : R' \\
   {\left\{
    \Gamma, p : P, r : Q \to R' ; Gk
    \vdash^{\Sigma'}_{G( g ; k)}
    M^k_{\op} : R'
    \right\}}^{k \colon c \to b}_{\op \colon \oprtyp{P}{Q}{d \xto{g} c} \in \Sigma}
  \end{array}
}
\]
Note that the environment $\Delta$ contains only variables typed by primitive types.
If it contained a type $A \to B; f$,
the morphism $f$ is in $\cats$
while the morphism must be in another grading category $\cats'$
after a computation graded by $f$ is handled.
This is impossible in general.

We omit the morphism $k$ in $\op(p), r \mapsto_k M^k_{\op}$
and write simply $\op(p), r \mapsto M_{\op}$
when $M_{\op} = M^k_{\op}$ for all $k$.
This convention is used in Example \ref{ex:mutable-type}.

\section{Operational Semantics} \label{operational-semantics}
To define the operational semantics of \cateff{}, we need some auxiliary notions.
\begin{definition}[Evaluation context]
  Evaluation contexts $\ctx$ and $\ctxf$ are defined by the following BNF:
  \[
  \ctx_{\Sigma} ::= [ \, ]
  \mid \letin{x}{\ctx_{\Sigma}}{M_{\Sigma}}
  \]
  \[
  \ctxf_{\Sigma'} ::= [ \, ]
  \mid \letin{x}{\ctxf_{\Sigma'}}{M_{\Sigma'}}
  \mid \handle{\ctxf_{\Sigma}}{H_{\Sigma \Rightarrow \Sigma'}}
  \]
  We write
  $\Gamma \vdash \ctxf : A \ctxto B; G(b') \xto{f} a$
  if $\Gamma, x : A \vdash_f \ctxf[\ret_{b'} x] : B$ is derived
  where $G$ is a functor that corresponds to handlers in $\ctxf$
  and $x$ does not appear in $\ctxf$ as a free variable.
\end{definition}

\begin{figure}
  \centering
  \[
  \begin{array}{lrcl}
    \textsc{S-App} & (\lambda^f x . M) V & \rto & M[V/x] \\
    \textsc{S-Let} & \letin{x}{\ret_a V}{M} & \rto & M[V/x] \\
    \textsc{S-Proj} & \proj{\langle V_1, V_2 \rangle}{x}{y}{M} & \rto & M[V_1/x, V_2/y] \\
    \textsc{S-MatchLeft} & \match{(\inl V)}{x}{M_1}{y}{M_2} & \rto & M_1[V/x] \\
    \textsc{S-MatchRight} & \match{(\inr V)}{x}{M_1}{y}{M_2} & \rto & M_2[V/y] \\
    \textsc{S-HandleRet} & \handle{(\ret_a V)}{H} & \rto & M[V/x]\\
    \textsc{S-HandleOp} & \handle{\ctx[\op(V)]}{H} & \rto & N \\
    \textsc{S-Lift} & \ctxf[M] & \rto & \ctxf[M'] \quad \text{if $M \rto M'$}
  \end{array}
  \]
  where
  $H = \{ \ret_a x \mapsto M \} \cup {\{ \op(p), r \mapsto_k M^k_{\op} \}}^k_{\op}$
  in \textsc{S-HandleRet} and \textsc{S-HandleOp},
  and
  $N = M^{k\colon c \to b}_{\op}[V/p, \lambda^{Gk} y . \handle{\ctx[\ret_b y]}{H} /r]$
  in \textsc{S-HandleOp}.
  \caption{Small-step operational semantics}
  \label{small-step-operational-semantics}
\end{figure}

\Cref{small-step-operational-semantics} gives
the small-step operational semantics for \cateff{}.
It is based on \cite{HL18}.
The rules of operational semantics are usual except for \textsc{S-HandleOp}.
In \textsc{S-HandleOp}, the grading morphism plays an important role.
Let $\ctx$ be
$\letin{x_n}{ (\dots (\letin{x_1}{[\ ]}{M_1}) \dots) }{M_n}$
and $f_i$ be the grading morphism of $M_i$ for each $i = 1, \dots , n$.
Consider a term $\handle{ \ctx[\op(V)] }{ H }$.
The handler $H$ handles $\op(V)$ with the term $M^{f_1 ; \dots ; f_n}_{\op}$ in $H$.
For example, see Example \ref{ex:handler}.

The goal of the rest of this section is to show the progress lemma and the preservation lemma for \cateff{}.
We start by proving the following lemmas by induction on derivations.
\begin{lemma}[Substitution]
  If $\Gamma, x_1 : A_1, \dots , x_n : A_n \vdash_f M : B$ and
  $\Gamma \vdash V_i : A_i$ for each $i = 1 , \dots, n$,
  then $\Gamma \vdash_f M[V_1 / x_1 , \dots, V_n / x_n] : B$.
\end{lemma}

\begin{lemma}\label{lem:context-substitution}
  If $\Gamma \vdash \ctxf : A \ctxto B; b \xto{f} a$,
  $\Gamma \vdash_{g' \colon c' \to b'} M : A$,
  and $G(b') = b$,
  then $\Gamma \vdash_{G(g') ; f} \ctxf[M] : B$
\end{lemma}

\begin{lemma}\label{lem:context-type}
  If $\Gamma \vdash_h \ctxf[M] : B$,
  then we have
  $\Gamma \vdash_{g'} M : A$
  and $\Gamma \vdash \ctxf : A \ctxto B; f$
  satisfying $h = G(g'); f$.
\end{lemma}

\begin{lemma}[Progress]\label{lem:progress}
  If $\vdash_{f \colon b \to a} M : A$
  then one of the following holds.
  \begin{enumerate}
  \item $f = \idmorph_a$ and $M = \ret_a V$ for some value term $V$,
  \item $M = \ctx[op(V)]$ for some $\ctx$, $\op$ and $V$, or
  \item there exists a computation term $N$ such that $M \rto N$.
  \end{enumerate}
\end{lemma}

\begin{lemma}[Preservation]\label{lem:preservation}
  If $\vdash_{f \colon b \to a} M : A$
  and $M \rto N$ then
  $\vdash_f N : B$ is derivable.
\end{lemma}

Together with the progress lemma(Lemma \ref{lem:progress}) and preservation lemma(Lemma \ref{lem:preservation}),
we have a safety theorem.
The safety theorem says that
if a computation term is well-typed then the term comes from value ($\ret_a V$) or is about to perform an operation.
\begin{theorem}[Safety]
  If $\vdash_{f \colon b \to a} M : A$ is a terminating term
  then
  there exists a value term $V$ such that $M = \ret_a V$,
  or $M$ calls an operation, that is $M = \ctx[\op(V)]$ for some $\ctx$, $\op \in \Sigma$ and $V$.
\end{theorem}

\begin{example}[Handler] \label{ex:handler}
  We present an example of small-step evaluation with handlers.
  Let $\cats$ be a category such that
  $\ob{\cats} = \{c, d, e\}$
  and morphisms are identities and $g \colon c \to d$ and $h \colon d \to e$,
  and $\cats'$ be a category such that
  $\ob{\cats'} = \{ \bullet \}$ and
  the identity is the only morphism.
  There is a unique functor $G \colon \cats \to \cats'$,
  which sends all objects in $\cats$ to $\bullet$ in $\cats'$.
  Let
  $\cats$-signature $\Sigma$ be
  $\{ \op_1 \colon \oprtyp{P}{A}{g}, \op_2 \colon \oprtyp{Q}{B}{h} \}$
  and
  $\cats'$-signature $\Sigma'$ be
  $\emptyset$.
  Given terms
  \[ \begin{array}{c}
  N = \letin{x}{\op_1(V)}{
    \letin{y}{\op_2(W)}{
      \ret_e \pair{x}{y}
    }
  }, \\
  M_{\op_1}^h = rV', \quad
  M_{\op_2}^{\idmorph_e} = rW', \\
  H = \{ \ret_e z \mapsto \ret_{\bullet} z; \
  \op_1(p), r \mapsto_h M_{\op_1}^h; \
  \op_2(p), r \mapsto_{\idmorph_e} M_{\op_2}^{\idmorph_e} \}
  \end{array} \]
  such that
  $\vdash V : P$, $\vdash W : Q$,
  $\vdash V' : A$, $\vdash W' : B$,
  $\vdash^{\Sigma}_{g; h} N : A \times B$
  and
  $\vdash^G_e H : A \times B \handlerto A \times B$.
  By the handler $H$,
  $\op_1$ and $\op_2$ are implemented as constant operations
  that always return values $V'$ and $W'$ for any arguments, respectively.
  By \textsc{Tc-Handle}, we have a judgement
  $\vdash^{\Sigma'}_{\idmorph_{\bullet}} \handle{N}{H} : A \times B$.
  Therefore, the term $\handle{N}{H}$ is evaluated to the form
  $\ret_{\bullet} U$ by the safety theorem.
  Indeed, $\handle{N}{H}$ is evaluated as follows:
  \[
  \begin{split}
    & \handle{N}{H} \\
    & \rto (
    \lambda^{Gh} v. \handle{
      (\letin{x}{\ret_d v}{
        \letin{y}{\op_2(W)}{
          \ret_e \pair{x}{y}
        }
      })
    }{H} ) V' \\
    & \rto^{*}
    \handle{
      (\letin{y}{\op_2(W)}{
        \ret_e \pair{V'}{y}
      })
    }{H} \\
    & \rto M^h_{\op_2}[W / p, \lambda^{G\idmorph_e} w.
      \handle{
        \letin{y}{\ret_e w}{
          \ret_e \pair{V'}{y}
        }
      }{H} / r] \\
    & =
    (\lambda^{G\idmorph_e} w.
    \handle{
      \letin{y}{\ret_e w}{
        \ret_e \pair{a}{y}
      }
    }{H}) W'
    \rto^{*}
    \ret_{\bullet} \pair{V'}{W'}.
  \end{split}
  \]
  We have
  $\vdash^{\Sigma'}_{\idmorph_{\bullet}} \ret_{\bullet} \pair{V'}{W'} : A \times B$.
\end{example}

\begin{example}[Mutable store of mutable type with a plan]\label{ex:mutable-type}
  We can make a program with a mutable store of mutable type with a mutation plan.
  Let $A$ and $B$ be primitive types,
  $V$ and $W$ are value terms such that $\vdash V : A$ and $\vdash W : B$.
  Let $\cats$ be a category such that
  $\ob{\cats} = \{\unit, A, B \}$ and
  morphisms are generated by
  $f^{\alpha}_{\beta} \colon \alpha \to \beta$ for
  $\alpha, \beta \in \ob{\cats}$
  and $\cats'$ be a category such that
  $\ob{\cats'} = \{\unit +  (A + B) \}$ and
  morphism is only identity.
  There is a unique functor $G \colon \cats \to \cats'$,
  which sends all objects in $\cats$ to $\unit +  (A + B)$ in $\cats'$.
  Intuitively, objects are possible types of mutable store
  and morphisms are plans of mutations of types of the store.
  Let $\Sigma$ be an $\cats$-signature
  ${\{ \update{\alpha}{\beta} \}}_{\alpha, \beta \in \ob{\cats}} \cup {\{ \lookup{\alpha} \}}_{\alpha \in \ob{\cats}}$
  where the type of $\update{\alpha}{\beta} \in \Sigma$ is $\oprtyp{\beta}{\unit}{f^{\alpha}_{\beta} \colon \alpha \to \beta}$
  and the type of $\lookup{\alpha}$ is $\oprtyp{\unit}{\alpha}{\idmorph_{\alpha}}$ for each $\alpha, \beta \in \ob{\cats}$.
  Let $\Sigma'$ be an $\cats'$-signature
  $\{ \update{}{}, \lookup{}{} \}$
  where the type of $\update{}{}$ is $\oprtyp{\unit + (A + B)}{\unit}{\idmorph}$
  and the type of $\lookup{}{}$ is $\oprtyp{\unit}{\unit + (A + B)}{\idmorph}$.

  We can implement the behavior of mutable types by the following handlers:
  $H_{\unit} =
  \{ \ret_{\unit} x \mapsto \ret_{\unit + (A + B)} \inl x\} \cup H_{\op}$,
  $H_{A} =
  \{ \ret_{A} x \mapsto \ret_{\unit + (A + B)} \inr (\inl x) \} \cup H_{\op}$ and
  $H_{B} =
  \{ \ret_{B} x \mapsto \ret_{\unit + (A + B)} \inr (\inr x) \} \cup H_{\op}$ where
  \[
  \begin{split}
    H_{\op} = \{
    & \update{\unit}{\unit}(\_), r \mapsto
    \letin{\_}{\update{}{}(\inl \unitval)}{r \unitval} \\
    & \vdots \\
    & \update{B}{B}(b), r \mapsto
    \letin{\_}{\update{}{}( \inr (\inr b))}{r \unitval} \\
    & \lookup{\unit}{}(\_), r \mapsto
    \lt x \bind \lookup{}(\unitval) \lin \\
    & \hspace{2.3cm} \mc x \{
    \inl y . r \unitval;
    \inr y . \mc y \{ \inl a. r \unitval; \inr b . r \unitval \} \} \\
    & \lookup{A}{}(\_), r \mapsto
    \lt x \bind \lookup{}(\unitval) \lin \\
    & \hspace{2.3cm} \mc x \{
    \inl y . r V;
    \inr y . \mc y \{ \inl a. r a; \inr b . r V \} \} \\
    & \lookup{B}{}(\_), r \mapsto
    \lt x \bind \lookup{}(\unitval) \lin \\
    & \hspace{2.3cm} \mc x \{
    \inl y . r W;
    \inr y . \mc y \{ \inl a. r W; \inr b . r b \} \} \\
  \end{split}
  \]
  We have judgements
  $\vdash^G_{\alpha} H_{\alpha} : \alpha \handlerto \unit + (A+B)$
  for all $\alpha \in \ob(\cats)$.
  Let us consider the term
  \[
  N = \lt \_ \bind \update{1}{A}(V) \lin
    \lt \_ \bind \update{A}{B}{W} \lin
    \lt x \bind \lookup{B}(\unitval) \lin
    \ret_B x
  \]
  We have a judgement $\vdash^{\Sigma}_{ f^{\unit}_A ; f^A_B } N : B$.
  The grading morphism $f^{\unit}_A ; f^A_B : \unit \to A \to B$
  implies that
  the program $N$ stores a value of type $A$
  and then stores a value of type $B$.
  Handling operations in $\Sigma$ by $H_{B}$,
  a mutable store of mutable type with a plan
  is transformed to a mutable store of fixed type $\unit + (A + B)$.
  We have
  $\vdash^{\Sigma'}_{\idmorph_{\unit + (A + B)}} \handle{N}{H_B} : \unit + (A + B)$.
\end{example}

\section{Denotational Semantics} \label{denotational-semantics}
In this section, we give denotational semantics for \cateff{}.
The denotational semantics is based on \cite{BP14}.
Let $\Sigma$, $\Sigma'$ be $\cats$- and $\cats'$-signatures, respectively.
For the sake of simplicity, we work in $\sets$.
To interpret computation terms that return a value in $X$,
we use the sets
$\Term{\Sigma}{f \colon b \to a}{X} = U_b \GATMod{\Sigma}{f}{\sets} F_a X$
defined in \cref{category-graded-algebraic-theory}.
We don't care about equations and so consider free models without any equations.
Each computational term is interpreted by a term of a category-graded algebraic theory.
We can think of elements of $\Term{\Sigma}{f \colon b \to a}{X}$
as trees whose leaves are labelled by elements of $X$,
and internal nodes are labelled by operations and their arguments.
For example,
$\pure(a, v) \in \Term{\Sigma}{\idmorph_a \colon a \to a}{X}$
is a tree with only one node labelled by $v$.
Note that, for any tree $t \in \Term{\Sigma}{f}{X}$,
we obtain $f$ by composing morphisms of label $\op$ of nodes
in a path from the root to a leaf.

\begin{definition}
  We define the interpretation of types $\denote{A} \in \sets$ as follows:
  \[ \begin{array}{cc}
  \denote{\unit} = \{ \star \}, \quad
  \denote{A \to B; b \xto{f} a} = \Term{\Sigma}{f}{\denote{B}}^{\denote{A}}, \\
  \denote{A + B} = \denote{A} \sqcup \denote{B}, \quad
  \denote{A \times B} = \denote{A} \times \denote{B}.
  \end{array} \]
  For a context $\Gamma = x_1 : A_1, \dots, x_n : A_n$,
  we define
  $\denote{\Gamma} = \denote{A_1} \times \dots \times \denote{A_n}$.
\end{definition}

We interpret
value $\Gamma \vdash V : A$,
computation $\Gamma \vdash_f M : A$
and handler $\Gamma \vdash^G_a H : R \handlerto R'$
as maps
$\denote{\Gamma \vdash V : A} \colon \denote{\Gamma} \to \denote{A}$,
$\denote{\Gamma \vdash_f M : A} \colon \denote{\Gamma} \to \Term{\Sigma}{f}{\denote{A}}$
and
$\denote{\Gamma \vdash^G_a H : R \handlerto R'}
\colon
\denote{\Gamma} \to {\Term{\Sigma'}{Gf}{\denote{R'}}}^{\Term{\Sigma}{f}{\denote{R}}}$,
respectively.
We write the lift of a morphism
$\phi \colon X \to \Term{\Sigma}{g}{Y}$ as
$\phi^{\dagger_f} \defeq \mu_{f,g} (\Term{\Sigma}{f}{\phi})$
and often simply write $\phi^{\dagger}$ when $f$ is clear from the context.

\begin{definition}
  Let $\Gamma$ be a typing context.
  Given $s \in \denote{\Gamma}$,
  we define the interpretation of terms as follows:
  \[ \begin{array}{c}
    \denote{\Gamma \vdash \unitval : \unit}s = \star, \quad
    \denote{\Gamma \vdash \pair{V_1}{V_2} : A \times B}s
    = \pair{\denote{V_1}s}{\denote{V_2}s}, \\
    \denote{\Gamma \vdash \inl{V} : A + B}s
    = \inj_1 {\denote{V}s}, \quad
    \denote{\Gamma \vdash \inr{V} : A + B}s
    = \inj_2 {\denote{V}s}, \\
    \denote{\Gamma, x : A \vdash x : A}s = \pi_x(s), \quad
    \denote{\Gamma \vdash \lambda^f x. M : A \to B; f}s
    = \denote{M}(s, \blank),
  \end{array} \]
  \[ \begin{array}{c}
    \denote{\Gamma \vdash_{\idmorph_a} \ret_a V : A}s
    = \pure(a, \denote{V}s), \quad
    \denote{\Gamma \vdash_f VW : B}s
    = (\denote{V}s)(\denote{W}s),
    \\
    \denote{\Gamma \vdash_f \op(V) : A}s
    = \doop(\op, \denote{V}s, {\{\pure(a, x)\}}_{x \in \denote{A}}),
    \\
    \denote{\Gamma \vdash_{f ; g} \letin{x}{M}{N} : B}s
    = {(\denote{N}(s, \blank))}^{\dagger} (\denote{M}s),
    \\
    \denote{\Gamma \vdash_f \proj{V}{x}{y}{M} : B}s
    = \denote{M}(s, \pi_1 (\denote{V}s), \pi_2 (\denote{V}s)),
    \\
    \denote{\Gamma \vdash_f \match{V}{x}{y}{M_1}{M_2} : B}s
    = [\denote{M_1}(s, \blank), \denote{M_2}(s, \blank)](\denote{V}s),
    \\
    \denote{\Gamma \vdash^{\Sigma'}_{Gf} \handle{M}{H} : R'}s
    = (\denote{H}s)(\denote{M}s),
  \end{array} \]
  \[
  \denote{\Gamma \vdash^G_a H : R \handlerto R'}s
  =
  \left\{
  \begin{array}{lcl}
    \pure(a, v) & \mapsto & \denote{M}(s, v) \\
    \doop(\op, p, {\{t_i\}}_i) & \mapsto & \denote{M^k_{\op}}(s, p, {\{\denote{H}(s, t_i)\}}_i)
  \end{array}
  \right.
  \]
  where
  $H = \{ \ret_a x \mapsto M \} \cup {\{\op(p), r \mapsto M^k_{\op} \}}^k_{\op \in \Sigma}$.
\end{definition}

\section{Soundness and Adequacy} \label{soundness-and-adequacy}
We show soundness and adequacy.
These theorems assert that
operational semantics and denotational semantics are compatible with each other.

\begin{theorem}[Soundness]
  If $\vdash_f M : A$ and $M \rto M'$
  then $\denote{\vdash_f M : A} = \denote{\vdash_f M' : A}$.
\end{theorem}

To show adequacy(Theorem \ref{thm:adequacy}),
we define relations $\apx_A$ for values and $\apx^f_A$ for computations
as done in \cite{BP14}.

\begin{definition}
  For $v \in \denote{A}$ and a closed value term $\vdash V : A$,
  we define $v \apx_A V$ as follows:
  \begin{itemize}
  \item $v \apx_{\unit} V$ if $v = \unit$ and $V = \unitval$.
  \item $v \apx_{A \times B} V$
    if $v = \pair{v_1}{v_2}$, $V = \pair{V_1}{V_2}$,
    $v_1 \apx_A V_1$ and $v_2 \apx_B V_2$.
  \item $v \apx_{A + B} V$ if either
    $v = \inj_1 v_1$, $V = \inl V_1$ and $v_1 \apx_A V_1$, or
    $v = \inj_2 v_2$, $V = \inr V_2$ and $v_2 \apx_B V_2$.
  \item $v \apx_{A \to B;f} V$
    if $v(w) \apx^f_B VW$
    for each $w \in \denote{A}$ and closed value $\vdash W : A$ satisfying $w \apx_A W$.
  \end{itemize}
  Simultaneously,
  for $c \in \Term{\Sigma}{f}{\denote{A}}$
  and a closed computation term $\vdash_f M : A$,
  $c \apx^f_A M$ holds if
  \begin{enumerate}
  \item $f = \idmorph_a$, $c = \e(a, v)$, $M \rto^* \ret_a V$ and $v \apx_A V$, or
  \item $c = \doop(\sigma, v, {\{t_x\}}_{x \in \denote{C}})$,
    $M \rto^* \ctx[\op(V)]$,
    $v \apx_P V$,
    and if $w \apx_C W$ then $t_w \apx^k_A \ctx[\ret_{b} W]$.
  \end{enumerate}
\end{definition}

\begin{lemma}\label{lem:apx-rev}
  If $c \apx^f_A M'$ and $M \rto M'$ then $c \apx^f_A M$.
\end{lemma}
\begin{proof}
  By assumption $c \apx^f_A M'$, we have two cases:
  \begin{enumerate}
  \item $f = \idmorph_a$, $c = \e(a, v)$, $M' \rto^* \ret_a V$ and $v \apx_A V$.
    In this case, we have $M \rto^* \ret_a V$.
  \item $c = \doop(\sigma, v, {\{t_x\}}_{x \in \denote{C}})$,
    $M' \rto^* \ctx[\op(V)]$,
    $v \apx_P V$,
    and if $w \apx_C W$ then $t_w \apx^k_A \ctx[\ret_{b} W]$.
    In this case, we have $M \rto^* \ctx[\op(V)]$.
  \end{enumerate}
  In both cases, we can conclude $c \apx^f_A M$ as required.
\end{proof}

\begin{lemma}\label{lem:denotation-apx-term}
  Let $\Gamma$ be a typing context $x_1 : A_1, \dots , x_n : A_n$, and
  $\vdash W_i : A_i$ be a closed value term
  and $w_i$ be an element of $\denote{A_i}$
  with $w_i \apx_{A_i} W_i$
  for each $1 \le i \le n$.
  Then followings hold.
  \begin{enumerate}
  \item If $\Gamma \vdash V : A$
    then $\denote{V}(w_1, \dots, w_n) \apx_A V[W_1/x_1, \dots, W_n/x_n]$.
  \item If $\Gamma \vdash_f M : A$
    then $\denote{M}(w_1, \dots, w_n) \apx^f_A M[W_1/x_1, \dots, W_n/x_n]$.
  \end{enumerate}
\end{lemma}

\begin{theorem}[Adequacy]\label{thm:adequacy}
  If $\vdash_{\idmorph_a} M : \unit$ and
  $\denote{\vdash_{\idmorph_a} M : \unit} = \pure(a, \star)$
  then
  $M \rto^{*} \ret_a \unitval$.
\end{theorem}
\begin{proof}
  By Lemma \ref{lem:denotation-apx-term},
  we have $\denote{M} \apx^{\idmorph_a}_{\unit} M$.
  Thus, we have $\pure(a,\star) \apx^{\idmorph_a}_{\unit} M$ by assumption.
  By the definition of $\apx^{\idmorph_a}_{\unit}$,
  we obtain $M \rto^* \ret_{a} V$
  and $\star \apx_{\unit} V$.
  Therefore, by definition of $\apx_{\unit}$,
  we conclude $V = \unitval$ and $M \rto^* \ret_{a} \unitval$.
\end{proof}

\section{Future Work}
\textbf{User-defined grading category and graded operations.}
The grading categories of \cateff{}
are considered to be built-in in this paper.
To provide user-defined grading operations,
the syntax to write grading categories and graded operations is needed.
This is future work.

\textbf{Category-graded Lawvere theories.}
Graded Lawvere theories
which correspond to graded algebraic theories
were developed in \cite{Kura20}.
Category-graded extensions of Lawvere theories are future work.

\textbf{2-category-graded monads.}
Graded monads are graded by partially ordered monoids in general.
We must consider \textit{2-category-graded monads},
whose grading category is 2-category,
to generalise partially ordered monoid-graded monads.
This situation is beyond \cite{Street72}.
Thus, the Eilenberg-Moore and Kleisli constructions on these monads are not trivial.

\section*{Acknowledgement}
I would like to thank the people of computer science group at RIMS,
  especially Masahito Hasegawa and Soichiro Fujii for discussions and comments,
  and the anonymous reviewers for comments.
  This work was supported by
  JST, the establishment of university fellowships towards the creation of science technology innovation, Grant Number JPMJFS2123
  and JST ERATO Grant Number JPMJER1603.

\bibliographystyle{entics}
\bibliography{cat-graded}

\appendix

\section{Generalised Units and Generalised Counits}\label{generalised-units}
\textit{Generalised units}
were introduced to unify category-graded monads and parameterised monads
in \cite{OWE20}.
In this section,
we introduce \textit{generalised counits of adjunctions}
that correspond to generalised units of monads
and investigate the role of generalised units in \cateff{}.

\subsection{Generalised Units of Category-Graded Monads}
First, we review the definition of
categroy-graded monads with generalised units.
\begin{definition}[Category-graded monads with generalised unit \cite{OWE20}]
  A \textit{category-graded monads with generalised unit} is
  a category-graded monad $(T \colon \cats^{\op} \laxto \Endo{\catc}, \eta, \mu)$
  together with the following data:
  \begin{itemize}
  \item A \textit{wide subcategory} $\catr$ of $\cats$, that is
    a subcategory $\catr \subseteq \cats$ satisfying $\ob{\catr} = \ob{\cats}$.
    We denote the inclusion functor by $\iota : \catr \hookrightarrow \cats$
    and usually omit $\iota$ when no confusion arises.
  \item For each morphism $f$ in $\catr$ and object $C$ in $\catc$,
    a morphism ${(\gu_{f})}_C \colon C \to T_f C$,
    satisfying the following commutative diagrams.
  \end{itemize}
  \[
  \begin{tikzcd}
    C
    \ar[r, "{(\gu_f)}_C"]
    \ar[d, "{(\gu_{g \circ f})}_C"']
    & T_f C
    \ar[d, "T_f {(\gu_g)}_C"] \\
    T_{g \circ f} C
    &
    T_f T_g C
    \ar[l, "\mu_{f,g}"]
  \end{tikzcd}
  \quad
  \begin{tikzcd}
    C
    \ar[r, "{(\eta_a)}_C"]
    \ar[rd, "{(\gu_{\idmorph_a})}_C"']
    & T_{\idmorph_a} C
    \ar[d, equal] \\
    & T_{\idmorph_a} C
  \end{tikzcd}
  \]
\end{definition}

Let $G \colon \cats^{\opposite} \to \cat$ be a functor
and $(J_a \colon \catc \to Ga, Ea \colon Ga \to \catc, \eta_a, \varepsilon_a)$
be adjunctions for each $a \in \ob{\cats}$ such that
$T_f = E_b Gf Ja$ for every $f \colon a \to b$ in $\cats$.
We depict the generalised units by string diagrams as follows.
\[
\begin{tikzcd}
  \catc
  \ar[r, "J_a"]
  \ar[d, "\idfunc_{\catc}"', ""{name=id, right}]
  & Ga
  \ar[d, "Gf", ""{name=Tf, left}] \\
  \catc
  & Gb
  \ar[l, "E_b"]
  \ar[from=id, to=Tf, Rightarrow, "\gu_{f}"]
\end{tikzcd}
=
\begin{tikzpicture}[baseline=0.5cm, x=0.5cm, y=0.5cm]
\coordinate (etaf) at (2.000000, 1.000000);
\coordinate (eb) at (3.000000, 2.000000);
\coordinate (gf) at (2.000000, 2.000000);
\coordinate (ja) at (1.000000, 2.000000);
\coordinate (rt) at (4.000000, 2.000000);
\coordinate (rb) at (4.000000, 0.000000);
\coordinate (lt) at (0.000000, 2.000000);
\coordinate (lb) at (0.000000, 0.000000);
\fill[black!20] (eb) to[out=-90, in=0] (etaf)  to[out=180, in=-90] (ja) -- (lt) -- (lb) -- (rb) -- (rt)-- cycle;
\fill[blue!50] (gf) -- (etaf)  to[out=0, in=-90] (eb)-- cycle;
\fill[red!50] (ja) to[out=-90, in=180] (etaf)  -- (gf)-- cycle;
\draw (gf) -- (etaf);
\draw (eb) to[out=-90, in=0] (etaf);
\draw (ja) to[out=-90, in=180] (etaf);
\node[above] at (ja) {$J_a$};
\node[above] at (eb) {$E_b$};
\node[above] at (gf) {$Gf$};
\node[below] at (etaf) {$\gu_f$};
\end{tikzpicture}
\]
Then, the above two rules are depicted as follows.
\[
\begin{tikzpicture}[baseline=0.5cm, x=0.5cm, y=0.5cm]
\coordinate (rb) at (4.000000, 0.000000);
\coordinate (rt) at (4.000000, 3.000000);
\coordinate (auxb) at (3.500000, 1.000000);
\coordinate (ug) at (3.000000, 0.500000);
\coordinate (auxa) at (0.500000, 1.000000);
\coordinate (eps) at (2.000000, 1.500000);
\coordinate (uf) at (1.000000, 0.500000);
\coordinate (comp) at (2.000000, 2.250000);
\coordinate (Ec) at (3.000000, 3.000000);
\coordinate (Ggf) at (2.000000, 3.000000);
\coordinate (Ja) at (1.000000, 3.000000);
\coordinate (lt) at (0.000000, 3.000000);
\coordinate (lb) at (0.000000, 0.000000);
\fill[black!20] (eps) to[out=180, in=0] (uf)  to[out=180, in=-90] (auxa) to[out=90, in=-90] (Ja) -- (lt) -- (lb) -- (rb) -- (rt) -- (Ec) to[out=-90, in=90] (auxb) to[out=-90, in=0] (ug) to[out=180, in=0] (eps)-- cycle;
\fill[green!50] (Ggf) -- (comp)  to[out=0, in=90] (ug) to[out=0, in=-90] (auxb) to[out=90, in=-90] (Ec)-- cycle;
\fill[blue!50] (comp) to[out=180, in=90] (uf)  to[out=0, in=180] (eps) to[out=0, in=180] (ug) to[out=90, in=0] (comp)-- cycle;
\fill[red!50] (Ggf) -- (comp)  to[out=180, in=90] (uf) to[out=180, in=-90] (auxa) to[out=90, in=-90] (Ja)-- cycle;
\draw (comp) to[out=0, in=90] (ug);
\draw (auxb) to[out=90, in=-90] (Ec);
\draw (ug) to[out=0, in=-90] (auxb);
\draw (ug) to[out=180, in=0] (eps);
\draw (comp) to[out=180, in=90] (uf);
\draw (uf) to[out=0, in=180] (eps);
\draw (auxa) to[out=90, in=-90] (Ja);
\draw (uf) to[out=180, in=-90] (auxa);
\draw (Ggf) -- (comp);

\node[above=1em] at (Ggf) {$G(g \circ f)$};
\node[above] at (Ja) {$J_a$};
\node[above] at (Ec) {$E_c$};
\end{tikzpicture}
=
\begin{tikzpicture}[baseline=0.5cm, x=0.5cm, y=0.5cm]
\coordinate (rb) at (4.000000, 0.000000);
\coordinate (rt) at (4.000000, 3.000000);
\coordinate (Ggf) at (2.000000, 3.000000);
\coordinate (Ec) at (3.000000, 3.000000);
\coordinate (ugf) at (2.000000, 1.000000);
\coordinate (Ja) at (1.000000, 3.000000);
\coordinate (lt) at (0.000000, 3.000000);
\coordinate (lb) at (0.000000, 0.000000);
\fill[black!20] (ugf) to[out=180, in=-90] (Ja)  -- (lt) -- (lb) -- (rb) -- (rt) -- (Ec) to[out=-90, in=0] (ugf)-- cycle;
\fill[green!50] (Ggf) -- (ugf)  to[out=0, in=-90] (Ec)-- cycle;
\fill[red!50] (Ggf) -- (ugf)  to[out=180, in=-90] (Ja)-- cycle;
\draw (ugf) to[out=0, in=-90] (Ec);
\draw (Ggf) -- (ugf);
\draw (ugf) to[out=180, in=-90] (Ja);

\node[above=1em] at (Ggf) {$G(g \circ f)$};
\node[above] at (Ja) {$J_a$};
\node[above] at (Ec) {$E_c$};
\end{tikzpicture}
\hspace{1cm}
\begin{tikzpicture}[baseline=0.5cm, x=0.5cm, y=0.5cm]
\coordinate (rb) at (4.000000, 0.000000);
\coordinate (rt) at (4.000000, 2.000000);
\coordinate (u) at (2.000000, 1.000000);
\coordinate (Ea) at (3.000000, 2.000000);
\coordinate (id) at (2.000000, 2.000000);
\coordinate (Ja) at (1.000000, 2.000000);
\coordinate (lt) at (0.000000, 2.000000);
\coordinate (lb) at (0.000000, 0.000000);
\fill[black!20] (u) to[out=180, in=-90] (Ja)  -- (lt) -- (lb) -- (rb) -- (rt) -- (Ea) to[out=-90, in=0] (u)-- cycle;
\fill[red!50] (id) -- (u)  to[out=0, in=-90] (Ea)-- cycle;
\fill[red!50] (id) -- (u)  to[out=180, in=-90] (Ja)-- cycle;
\draw (u) to[out=0, in=-90] (Ea);
\draw (id) -- (u);
\draw (u) to[out=180, in=-90] (Ja);

\node[above=1em] at (id) {$\idmorph_{G_a}$};
\node[above] at (Ja) {$J_a$};
\node[above] at (Ea) {$E_a$};
\node[below] at (u) {$\gu_{\idmorph_a}$};
\end{tikzpicture}
=
\begin{tikzpicture}[baseline=0.5cm, x=0.5cm, y=0.5cm]
\coordinate (rb) at (4.000000, 0.000000);
\coordinate (rt) at (4.000000, 2.000000);
\coordinate (u) at (2.000000, 1.000000);
\coordinate (Ea) at (3.000000, 2.000000);
\coordinate (id) at (2.000000, 2.000000);
\coordinate (Ja) at (1.000000, 2.000000);
\coordinate (lt) at (0.000000, 2.000000);
\coordinate (lb) at (0.000000, 0.000000);
\fill[red!50] (Ea) to[out=-90, in=0] (u)  to[out=180, in=-90] (Ja)-- cycle;
\fill[black!20] (u) to[out=180, in=-90] (Ja)  -- (lt) -- (lb) -- (rb) -- (rt) -- (Ea) to[out=-90, in=0] (u)-- cycle;
\draw (u) to[out=0, in=-90] (Ea);
\draw (u) to[out=180, in=-90] (Ja);

\node[above] at (Ja) {$J_a$};
\node[above] at (Ea) {$E_a$};
\node[below] at (u) {$\eta_{a}$};
\end{tikzpicture}
\]

A category-graded monad with generalised unit
whose wide subcategory is a discrete category
is an ordinary parameterised monad.
Every parameterised monad can be seen
a category-graded monad with generalised units
using \textit{pair completion}, see \cite{OWE20} and \Cref{correspondence-of-p-monads-to-cat-graded-monads-with-gu}.

\subsection{Generalised Counits of Adjunctions}
We introduce \textit{generalised counits} of adjunctions
and show that generalised units of monads correspond to generalised counits of adjunctions.
\begin{definition}[Generalised counits of adjunctions]
Let $G \colon \cats^{\op} \to \cat$ be a functor and
$(J_a \colon \cats \to Ga, E_a \colon Ga \to \cats, \eta_a, \varepsilon_a)$ be adjunctions.
A \textit{generalised counit} of those adjunctions
consists of the following data.
\begin{itemize}
\item A wide subcategory $\catr$ of $\cats$.
\item For each morphism $f \colon b \to a$ in $\catr$,
  a natural transformation $\gcu_f \colon J_b E_a \nattr Gf$,
  satisfying the following commutative diagrams.
\end{itemize}
\[
\begin{tikzcd}
  J_c E_a A
  \ar[r, "J_c {(\eta_b)}_{E_a A}"]
  \ar[d, "{(\gcu_{f \circ g})}_A"]
  & J_c E_b J_b E_a A
  \ar[r, "J_c E_b {(\gcu_g)}_A"]
  & J_c E_b Gg A
  \ar[d, "{(\gcu_f)}_{GgA}"] \\
  G(f \circ g) A
  \ar[rr, equal]
  &
  & GfGg A
\end{tikzcd}
\begin{tikzcd}
  J_a E_a A
  \ar[r, "{(\varepsilon_a)}_A"]
  \ar[rd, "{(\gcu_{\idmorph_a})}_A"']
  & A
  \ar[d, equal] \\
  & A
\end{tikzcd}
\]
\end{definition}
We depict generalised counits by string diagrams as follows:
\[
\begin{tikzcd}
  \catc
  \ar[d, "\idfunc_{\catc}"', ""{name=id, right}]
  & Ga
  \ar[l, "E_a"']
  \ar[d, "Gf", ""{name=Tf, left}] \\
  \catc
  \ar[r, "J_b"]
  & Gb
  \ar[from=id, to=Tf, Rightarrow, "\gcu_{f}"]
\end{tikzcd}
=
\begin{tikzpicture}[baseline=0.5cm, x=0.5cm, y=0.5cm]
\coordinate (jb) at (3.000000, 0.000000);
\coordinate (ea) at (1.000000, 0.000000);
\coordinate (epsf) at (2.000000, 1.000000);
\coordinate (gf) at (2.000000, 2.000000);
\coordinate (rt) at (4.000000, 2.000000);
\coordinate (rb) at (4.000000, 0.000000);
\coordinate (lt) at (0.000000, 2.000000);
\coordinate (lb) at (0.000000, 0.000000);
\fill[black!20] (ea) to[out=90, in=180] (epsf)  to[out=0, in=90] (jb)-- cycle;
\fill[blue!50] (jb) to[out=90, in=0] (epsf)  -- (gf) -- (rt) -- (rb)-- cycle;
\fill[red!50] (ea) to[out=90, in=180] (epsf)  -- (gf) -- (lt) -- (lb)-- cycle;
\draw (epsf) -- (gf);
\draw (jb) to[out=90, in=0] (epsf);
\draw (ea) to[out=90, in=180] (epsf);

\node[above] at (gf) {$Gf$};
\node[below] at (epsf) {$\gcu_f$};
\node[below] at (jb) {$J_b$};
\node[below] at (ea) {$E_a$};
\end{tikzpicture}
\]

Then the above two commutative diagrams are
\[ \begin{tikzpicture}[baseline=0.5cm, x=0.5cm, y=0.5cm]
\coordinate (rb) at (5.000000, 0.000000);
\coordinate (rt) at (5.000000, 3.000000);
\coordinate (etab) at (2.500000, 0.500000);
\coordinate (lb) at (0.000000, 0.000000);
\coordinate (lt) at (0.000000, 3.000000);
\coordinate (eg) at (3.500000, 1.000000);
\coordinate (ef) at (1.500000, 1.000000);
\coordinate (Jc) at (4.000000, 0.000000);
\coordinate (Ea) at (1.000000, 0.000000);
\coordinate (c) at (2.500000, 2.000000);
\coordinate (Ggf) at (2.500000, 3.000000);
\fill[green!50] (rt) -- (rb)  -- (Jc) to[out=90, in=0] (eg) to[out=90, in=0] (c) -- (Ggf)-- cycle;
\fill[blue!50] (etab) to[out=180, in=0] (ef)  to[out=90, in=180] (c) to[out=0, in=90] (eg) to[out=180, in=0] (etab)-- cycle;
\fill[red!50] (lt) -- (lb)  -- (Ea) to[out=90, in=180] (ef) to[out=90, in=180] (c) -- (Ggf)-- cycle;
\fill[black!20] (Ea) to[out=90, in=180] (ef)  to[out=0, in=180] (etab) to[out=0, in=180] (eg) to[out=0, in=90] (Jc)-- cycle;
\draw (c) -- (Ggf);
\draw (etab) to[out=0, in=180] (eg);
\draw (eg) to[out=90, in=0] (c);
\draw (ef) to[out=90, in=180] (c);
\draw (etab) to[out=180, in=0] (ef);
\draw (Ea) to[out=90, in=180] (ef);
\draw (Jc) to[out=90, in=0] (eg);

\node[above] at (Ggf) {$G(g\circ f)$};
\node[below] at (Ea) {$E_a$};
\node[below] at (Jc) {$J_c$};
\end{tikzpicture}
=
\begin{tikzpicture}[baseline=0.5cm, x=0.5cm, y=0.5cm]
\coordinate (rb) at (5.000000, 0.000000);
\coordinate (rt) at (5.000000, 3.000000);
\coordinate (lb) at (0.000000, 0.000000);
\coordinate (lt) at (0.000000, 3.000000);
\coordinate (Jc) at (4.000000, 0.000000);
\coordinate (Ea) at (1.000000, 0.000000);
\coordinate (comp) at (2.500000, 1.500000);
\coordinate (Ggf) at (2.500000, 3.000000);
\fill[black!20] (Ea) to[out=90, in=180] (comp)  to[out=0, in=90] (Jc)-- cycle;
\fill[green!50] (rt) -- (rb)  -- (Jc) to[out=90, in=0] (comp) -- (Ggf)-- cycle;
\fill[red!50] (lt) -- (lb)  -- (Ea) to[out=90, in=180] (comp) -- (Ggf)-- cycle;
\draw (Jc) to[out=90, in=0] (comp);
\draw (Ea) to[out=90, in=180] (comp);
\draw (comp) -- (Ggf);

\node[above] at (Ggf) {$G(g\circ f)$};
\node[below] at (Ea) {$E_a$};
\node[below] at (Jc) {$J_c$};
\end{tikzpicture} \qquad
\begin{tikzpicture}[baseline=0.5cm, x=0.5cm, y=0.5cm]
\coordinate (Ja) at (2.000000, 0.000000);
\coordinate (Ea) at (1.000000, 0.000000);
\coordinate (eps) at (1.500000, 1.000000);
\coordinate (id) at (1.500000, 2.000000);
\coordinate (rt) at (3.000000, 2.000000);
\coordinate (rb) at (3.000000, 0.000000);
\coordinate (lt) at (0.000000, 2.000000);
\coordinate (lb) at (0.000000, 0.000000);
\fill[black!20] (Ea) to[out=90, in=180] (eps)  to[out=0, in=90] (Ja)-- cycle;
\fill[red!50] (rt) -- (rb)  -- (Ja) to[out=90, in=0] (eps) -- (id)-- cycle;
\fill[red!50] (lt) -- (lb)  -- (Ea) to[out=90, in=180] (eps) -- (id)-- cycle;
\draw (Ja) to[out=90, in=0] (eps);
\draw (Ea) to[out=90, in=180] (eps);
\draw (eps) -- (id);

\node[above] at (id) {$\idmorph_{Ga}$};
\node[below] at (Ea) {$E_a$};
\node[below] at (Ja) {$J_a$};
\node[below] at (eps) {$\gcu_{\idmorph_a}$};
\end{tikzpicture}
=
\begin{tikzpicture}[baseline=0.5cm, x=0.5cm, y=0.5cm]
\coordinate (Ja) at (2.000000, 0.000000);
\coordinate (Ea) at (1.000000, 0.000000);
\coordinate (eps) at (1.500000, 1.000000);
\coordinate (id) at (1.500000, 2.000000);
\coordinate (rt) at (3.000000, 2.000000);
\coordinate (rb) at (3.000000, 0.000000);
\coordinate (lt) at (0.000000, 2.000000);
\coordinate (lb) at (0.000000, 0.000000);
\fill[red!50] (lt) -- (lb)  -- (Ea) to[out=90, in=180] (eps) to[out=0, in=90] (Ja) -- (rb) -- (rt)-- cycle;
\fill[black!20] (Ea) to[out=90, in=180] (eps)  to[out=0, in=90] (Ja)-- cycle;
\draw (Ja) to[out=90, in=0] (eps);
\draw (Ea) to[out=90, in=180] (eps);
\node[below] at (Ea) {$E_a$};
\node[below] at (Ja) {$J_a$};
\node[below] at (eps) {$\varepsilon_{a}$};
\end{tikzpicture}. \]

Next proposition is analogous to the usual unit law of monads.
\begin{proposition}
  The following equations hold.
  \[
\begin{tikzpicture}[baseline=0.5cm, x=0.5cm, y=0.5cm]
\coordinate (rt) at (6.000000, 4.000000);
\coordinate (rb) at (6.000000, 0.000000);
\coordinate (lb) at (0.000000, 0.000000);
\coordinate (lt) at (0.000000, 4.000000);
\coordinate (ebb) at (5.000000, 0.000000);
\coordinate (ebt) at (4.000000, 4.000000);
\coordinate (gf) at (3.000000, 4.000000);
\coordinate (jb) at (4.000000, 0.000000);
\coordinate (eps) at (2.750000, 2.500000);
\coordinate (etab) at (2.000000, 1.000000);
\coordinate (auxa) at (0.500000, 2.000000);
\coordinate (ja) at (2.000000, 4.000000);
\coordinate (epsf) at (1.500000, 2.000000);
\coordinate (etaa) at (1.000000, 1.000000);
\fill[black!20] (ebt) to[out=-90, in=90] (ebb)  -- (rb) -- (rt)-- cycle;
\fill[black!20] (lt) -- (lb)  -- (jb) to[out=90, in=0] (eps) to[out=180, in=0] (etab) to[out=180, in=0] (epsf) to[out=180, in=0] (etaa) to[out=180, in=-90] (auxa) to[out=90, in=-90] (ja)-- cycle;
\fill[blue!50] (jb) to[out=90, in=0] (eps)  to[out=180, in=0] (etab) to[out=180, in=0] (epsf) to[out=90, in=-90] (gf) -- (ebt) to[out=-90, in=90] (ebb)-- cycle;
\fill[red!50] (ja) to[out=-90, in=90] (auxa)  to[out=-90, in=180] (etaa) to[out=0, in=180] (epsf) to[out=90, in=-90] (gf)-- cycle;
\draw (ebt) to[out=-90, in=90] (ebb);
\draw (etab) to[out=180, in=0] (epsf);
\draw (eps) to[out=180, in=0] (etab);
\draw (jb) to[out=90, in=0] (eps);
\draw (epsf) to[out=90, in=-90] (gf);
\draw (etaa) to[out=0, in=180] (epsf);
\draw (auxa) to[out=-90, in=180] (etaa);
\draw (ja) to[out=-90, in=90] (auxa);

\node[above] at (ja) {$J_a$};
\node[above] at (ebt) {$E_b$};
\node[above] at (gf) {$Gf$};

\node[below] at (ebb) {$E_b$};
\node[below] at (jb) {$J_b$};
\end{tikzpicture}
=
\begin{tikzpicture}[baseline=0.5cm, x=0.5cm, y=0.5cm]
\coordinate (rt) at (4.000000, 4.000000);
\coordinate (rb) at (4.000000, 0.000000);
\coordinate (lb) at (0.000000, 0.000000);
\coordinate (lt) at (0.000000, 4.000000);
\coordinate (ebb) at (3.000000, 0.000000);
\coordinate (ebt) at (3.000000, 4.000000);
\coordinate (gf) at (2.000000, 4.000000);
\coordinate (jb) at (2.000000, 0.000000);
\coordinate (auxa) at (0.500000, 2.000000);
\coordinate (ja) at (1.000000, 4.000000);
\coordinate (epsf) at (1.500000, 2.000000);
\coordinate (etaa) at (1.000000, 1.000000);
\fill[red!50] (ja) to[out=-90, in=90] (auxa)  to[out=-90, in=180] (etaa) to[out=0, in=180] (epsf) to[out=90, in=-90] (gf)-- cycle;
\fill[blue!50] (ebt) -- (ebb)  -- (jb) to[out=90, in=0] (epsf) to[out=90, in=-90] (gf)-- cycle;
\fill[black!20] (lt) -- (lb)  -- (jb) to[out=90, in=0] (epsf) to[out=180, in=0] (etaa) to[out=180, in=-90] (auxa) to[out=90, in=-90] (ja)-- cycle;
\fill[black!20] (ebt) -- (ebb)  -- (rb) -- (rt)-- cycle;
\draw (epsf) to[out=90, in=-90] (gf);
\draw (etaa) to[out=0, in=180] (epsf);
\draw (auxa) to[out=-90, in=180] (etaa);
\draw (ja) to[out=-90, in=90] (auxa);
\draw (jb) to[out=90, in=0] (epsf);
\draw (ebt) -- (ebb);

\node[above] at (ja) {$J_a$};
\node[above] at (ebt) {$E_b$};
\node[above] at (gf) {$Gf$};
\node[below] at (ebb) {$E_b$};
\node[below] at (jb) {$J_b$};
\end{tikzpicture}
\hspace{0.5cm}
\begin{tikzpicture}[baseline=0.5cm, x=0.5cm, y=0.5cm]
\coordinate (auxb) at (3.500000, 2.000000);
\coordinate (eb) at (3.000000, 4.000000);
\coordinate (etab) at (3.000000, 1.000000);
\coordinate (epsf) at (2.500000, 2.000000);
\coordinate (ea) at (2.000000, 0.000000);
\coordinate (gf) at (2.000000, 4.000000);
\coordinate (jat) at (1.000000, 4.000000);
\coordinate (jab) at (1.000000, 0.000000);
\coordinate (rb) at (4.000000, 0.000000);
\coordinate (rt) at (4.000000, 4.000000);
\coordinate (lt) at (0.000000, 4.000000);
\coordinate (lb) at (0.000000, 0.000000);
\fill[blue!50] (eb) to[out=-90, in=90] (auxb)  to[out=-90, in=0] (etab) to[out=180, in=0] (epsf) to[out=90, in=-90] (gf)-- cycle;
\fill[black!20] (lb) -- (lt)  -- (jat) to[out=-90, in=90] (jab)-- cycle;
\fill[red!50] (jat) to[out=-90, in=90] (jab)  -- (ea) to[out=90, in=180] (epsf) to[out=90, in=-90] (gf)-- cycle;
\fill[black!20] (rb) -- (rt)  -- (eb) to[out=-90, in=90] (auxb) to[out=-90, in=0] (etab) to[out=180, in=0] (epsf) to[out=180, in=90] (ea)-- cycle;
\draw (ea) to[out=90, in=180] (epsf);
\draw (etab) to[out=180, in=0] (epsf);
\draw (auxb) to[out=-90, in=0] (etab);
\draw (eb) to[out=-90, in=90] (auxb);
\draw (epsf) to[out=90, in=-90] (gf);
\draw (jat) to[out=-90, in=90] (jab);

\node[above] at (jat) {$J_a$};
\node[above] at (gf) {$Gf$};
\node[above] at (eb) {$E_b$};

\node[below] at (jab) {$J_a$};
\node[below] at (ea) {$E_a$};
\end{tikzpicture}
=
\begin{tikzpicture}[baseline=0.5cm, x=0.5cm, y=0.5cm]
\coordinate (auxb) at (5.500000, 2.000000);
\coordinate (eb) at (4.000000, 4.000000);
\coordinate (etab) at (5.000000, 1.000000);
\coordinate (epsf) at (4.500000, 2.000000);
\coordinate (etaa) at (4.000000, 1.000000);
\coordinate (eps) at (3.250000, 2.500000);
\coordinate (ea) at (2.000000, 0.000000);
\coordinate (gf) at (3.000000, 4.000000);
\coordinate (jat) at (2.000000, 4.000000);
\coordinate (jab) at (1.000000, 0.000000);
\coordinate (rb) at (6.000000, 0.000000);
\coordinate (rt) at (6.000000, 4.000000);
\coordinate (lt) at (0.000000, 4.000000);
\coordinate (lb) at (0.000000, 0.000000);
\fill[black!20] (rb) -- (rt)  -- (eb) to[out=-90, in=90] (auxb) to[out=-90, in=0] (etab) to[out=180, in=0] (epsf) to[out=180, in=0] (etaa) to[out=180, in=0] (eps) to[out=180, in=90] (ea)-- cycle;
\fill[blue!50] (eb) to[out=-90, in=90] (auxb)  to[out=-90, in=0] (etab) to[out=180, in=0] (epsf) to[out=90, in=-90] (gf)-- cycle;
\fill[red!50] (jat) to[out=-90, in=90] (jab)  -- (ea) to[out=90, in=180] (eps) to[out=0, in=180] (etaa) to[out=0, in=180] (epsf) to[out=90, in=-90] (gf)-- cycle;
\fill[black!20] (lb) -- (lt)  -- (jat) to[out=-90, in=90] (jab)-- cycle;
\draw (etab) to[out=180, in=0] (epsf);
\draw (auxb) to[out=-90, in=0] (etab);
\draw (eb) to[out=-90, in=90] (auxb);
\draw (epsf) to[out=90, in=-90] (gf);
\draw (etaa) to[out=0, in=180] (epsf);
\draw (eps) to[out=0, in=180] (etaa);
\draw (ea) to[out=90, in=180] (eps);
\draw (jat) to[out=-90, in=90] (jab);

\node[above] at (jat) {$J_a$};
\node[above] at (gf) {$Gf$};
\node[above] at (eb) {$E_b$};

\node[below] at (jab) {$J_a$};
\node[below] at (ea) {$E_a$};
\end{tikzpicture}
\]
\end{proposition}
\begin{proof}
  We can easily check by deformation of strings.
\end{proof}

\begin{theorem}
  Let $G \colon \cats^{\op} \to \cat$ be a functor and
  $(J_a \colon \cats \to Ga, E_a \colon Ga \to \cats, \eta_a, \varepsilon_a)$
  be adjunctions.
  There is a one to one correspondence between
  generalised units of the category-graded monad induced by $G$ and $(J_a, E_a, \eta_a, \varepsilon_a)$,
  and generalised counits of the adjunctions.
\end{theorem}
\begin{proof}
  We can define a map from generalised counits to generalised units and its inverse by:
  \[
  \begin{tikzpicture}[baseline=0.5cm, x=0.5cm, y=0.5cm]
\coordinate (jb) at (3.000000, 0.000000);
\coordinate (ea) at (1.000000, 0.000000);
\coordinate (epsf) at (2.000000, 1.000000);
\coordinate (gf) at (2.000000, 2.000000);
\coordinate (rt) at (4.000000, 2.000000);
\coordinate (rb) at (4.000000, 0.000000);
\coordinate (lt) at (0.000000, 2.000000);
\coordinate (lb) at (0.000000, 0.000000);
\fill[black!20] (ea) to[out=90, in=180] (epsf)  to[out=0, in=90] (jb)-- cycle;
\fill[blue!50] (jb) to[out=90, in=0] (epsf)  -- (gf) -- (rt) -- (rb)-- cycle;
\fill[red!50] (ea) to[out=90, in=180] (epsf)  -- (gf) -- (lt) -- (lb)-- cycle;
\draw (epsf) -- (gf);
\draw (jb) to[out=90, in=0] (epsf);
\draw (ea) to[out=90, in=180] (epsf);

\node[above] at (gf) {$Gf$};
\node[below] at (epsf) {$\gcu_f$};
\node[below] at (jb) {$J_b$};
\node[below] at (ea) {$E_a$};
\end{tikzpicture}
  \mapsto
  \begin{tikzpicture}[baseline=0.5cm, x=0.5cm, y=0.5cm]
\coordinate (rt) at (6.000000, 3.000000);
\coordinate (rb) at (6.000000, 0.000000);
\coordinate (lt) at (0.000000, 3.000000);
\coordinate (lb) at (0.000000, 0.000000);
\coordinate (Eb) at (5.000000, 3.000000);
\coordinate (etab) at (4.000000, 1.000000);
\coordinate (Gf) at (3.000000, 3.000000);
\coordinate (cu) at (3.000000, 2.000000);
\coordinate (etaa) at (2.000000, 1.000000);
\coordinate (Ja) at (1.000000, 3.000000);
\fill[black!20] (cu) to[out=180, in=0] (etaa)  to[out=180, in=-90] (Ja) -- (lt) -- (lb) -- (rb) -- (rt) -- (Eb) to[out=-90, in=0] (etab) to[out=180, in=0] (cu)-- cycle;
\fill[blue!50] (Gf) -- (cu)  to[out=0, in=180] (etab) to[out=0, in=-90] (Eb)-- cycle;
\fill[red!50] (Gf) -- (cu)  to[out=180, in=0] (etaa) to[out=180, in=-90] (Ja)-- cycle;
\draw (etab) to[out=0, in=-90] (Eb);
\draw (cu) to[out=0, in=180] (etab);
\draw (etaa) to[out=180, in=-90] (Ja);
\draw (cu) to[out=180, in=0] (etaa);
\draw (Gf) -- (cu);

\node[above] at (Ja) {$J_a$};
\node[above] at (Eb) {$E_b$};
\node[above] at (Gf) {$Gf$};
\node[below] at (cu) {$\gcu_f$};
\end{tikzpicture}
  \quad \text{and} \quad
  \begin{tikzpicture}[baseline=0.5cm, x=0.5cm, y=0.5cm]
\coordinate (etaf) at (2.000000, 1.000000);
\coordinate (eb) at (3.000000, 2.000000);
\coordinate (gf) at (2.000000, 2.000000);
\coordinate (ja) at (1.000000, 2.000000);
\coordinate (rt) at (4.000000, 2.000000);
\coordinate (rb) at (4.000000, 0.000000);
\coordinate (lt) at (0.000000, 2.000000);
\coordinate (lb) at (0.000000, 0.000000);
\fill[black!20] (eb) to[out=-90, in=0] (etaf)  to[out=180, in=-90] (ja) -- (lt) -- (lb) -- (rb) -- (rt)-- cycle;
\fill[blue!50] (gf) -- (etaf)  to[out=0, in=-90] (eb)-- cycle;
\fill[red!50] (ja) to[out=-90, in=180] (etaf)  -- (gf)-- cycle;
\draw (gf) -- (etaf);
\draw (eb) to[out=-90, in=0] (etaf);
\draw (ja) to[out=-90, in=180] (etaf);
\node[above] at (ja) {$J_a$};
\node[above] at (eb) {$E_b$};
\node[above] at (gf) {$Gf$};
\node[below] at (etaf) {$\gu_f$};
\end{tikzpicture}
  \mapsto
  \begin{tikzpicture}[baseline=0.5cm, x=0.5cm, y=0.5cm]
\coordinate (gf) at (3.000000, 3.000000);
\coordinate (jb) at (5.000000, 0.000000);
\coordinate (epsb) at (4.000000, 2.000000);
\coordinate (epsa) at (2.000000, 2.000000);
\coordinate (ea) at (1.000000, 0.000000);
\coordinate (etaf) at (3.000000, 1.000000);
\coordinate (rt) at (6.000000, 3.000000);
\coordinate (rb) at (6.000000, 0.000000);
\coordinate (lt) at (0.000000, 3.000000);
\coordinate (lb) at (0.000000, 0.000000);
\fill[black!20] (ea) to[out=90, in=180] (epsa)  to[out=0, in=180] (etaf) to[out=0, in=180] (epsb) to[out=0, in=90] (jb)-- cycle;
\fill[blue!50] (jb) to[out=90, in=0] (epsb)  to[out=180, in=0] (etaf) -- (gf) -- (rt) -- (rb)-- cycle;
\fill[red!50] (ea) to[out=90, in=180] (epsa)  to[out=0, in=180] (etaf) -- (gf) -- (lt) -- (lb)-- cycle;
\draw (epsb) to[out=180, in=0] (etaf);
\draw (jb) to[out=90, in=0] (epsb);
\draw (epsa) to[out=0, in=180] (etaf);
\draw (ea) to[out=90, in=180] (epsa);
\draw (etaf) -- (gf);

\node[above] at (gf) {$Gf$};
\node[below] at (etaf) {$\gu_f$};
\node[below] at (ea) {$E_a$};
\node[below] at (jb) {$J_a$};
\end{tikzpicture}.
  \]
\end{proof}

\subsection{Generalised Units in \cateff}
We can introduce generalised units to \cateff{}.
We fix a grading category $\cats$ and wide subcategory $\catr$ of it.
The typing rule corresponding to generalised units is as follows.
\[
\infer[\textsc{Tc-Gunit}]{
  \Gamma \vdash_{h \circ f \circ g} M : A
}{
  \text{$g \colon a' \to a$ in $\catr$}
  &
  \Gamma \vdash_{f \colon a \to b} M : A
  &
  \text{$h \colon b \to b'$ in $\catr$}
}
\]
If we think of objects in grading category as conditions of states,
the rule \textsc{Tc-Gunit} represents weakening of
the condition along the grading morphism.
For example,
if objects in the grading category are types of states,
generalised units represent the subtyping relation of the types.
The denotation of a judgement derived by \textsc{Tc-Gunit} is defined by
$\denote{\Gamma \vdash_{h \circ f \circ g} M : A}s
= \mu_{f \circ g, h}
(\Term{\Sigma}{f \circ g}{\gu_h}
(\mu_{g,f}
(\gu_g (\denote{\Gamma \vdash_f M : A}s))))$.

\section{Correspondence of Parameterised Monads to Category-Graded Monads with Generalised Units}
\label{correspondence-of-p-monads-to-cat-graded-monads-with-gu}
In this section,
we describe correspondence of parameterised monads
to category-graded monads with generalised units.
This correspondence was shown by \cite{OWE20},
but it needs some modification.

\begin{definition}[Parameterised monad \cite{Atkey09}]
  A \textit{parameterised monad} consists of
  a functor $P \colon \cats^{\opposite} \times \cats \to [\catc, \catc]$,
  a natural transformation $\eta^{P}_{a} \colon \idfunc_{\catc} \nattr P(a,a)$ for each $a \in \ob{\cats}$,
  and a natural transformation $\mu^{P}_{a, b, c} \colon P(a,b) P(b,c) \nattr P(a,c)$ for each $a, b, c \in \ob{\cats}$
  satisfying
  appropriate commutative diagram,
  and $\eta^P_a$ is dinatural in $a$
  and $\mu^P_{a,b,c}$ is dinatural in $b$ and natural in $a,c$.
\end{definition}

We introduce the pair completion
to unify parameterised monads and category-graded monads with generalised units.
\begin{definition}[Pair completion \cite{OWE20}]
  Let $\cats$ be a small category.
  The \textit{pair completion} $\cats^{\nabla}$ of $\cats$ is
  a category whose objects are the objects of $\cats$
  and homsets are $\cats^{\nabla}(a,b) = \cats(a,b) \sqcup \{ (a,b) \}$.
  Composition of morphisms are defined as follows:
  \begin{equation*}
    \inj_1 g \circ \inj_1 f = \inj_1 (g \circ f),
    \quad
    \inj_1 g \circ \inj_2 g = \inj_2 (a,c),
    \quad
    \inj_2 (b,c) \circ \inj_1 f = \inj_2 (a,c),
    \quad
    \inj_2 (b,c) \circ \inj_2 (a,b) = \inj_2 (a,c)
  \end{equation*}
  where $a, b, c \in \ob{\cats}$,
  $f \colon a \to b$ and $g \colon b \to c$.
\end{definition}

The following two propositions show
the correspondence of parameterised monads
to category-graded monads with generalised units.
\begin{proposition}[Category-graded monads with generalised units from parameterised monads]
  Let $P \colon \cats^{\opposite} \times \cats \to [\catc, \catc]$
  be a parameterised monad.
  We define $(T_P, \eta^{T_P}, \mu^{T_P}, \gu^{T_P})$ as follows:
  \begin{itemize}
  \item $T_P(a) \defeq \catc$, $T_P(f) \defeq P(a,b)$ for $f \colon a \to b$ in $\cats^{\nabla}$,
  \item $\eta^{T_P}_{a} \defeq \eta^P_a
    \colon \idfunc_{\catc} \nattr T_P(\idmorph_a)$
    for $a \in \ob{\cats}$,
  \item $\mu^{T_P}_{f,g} \defeq \mu^P_{a,b,c}
    \colon T_P(f) \circ T_P(g) \nattr T_P(g \circ f)$
    for $a, b, c \in \ob{\cats}$,
    $f \colon a \to b$ and $g \colon b \to c$ in $\cats^{\nabla}$ and
  \item $\gu^{T_P}_{\inj_1 f} \defeq P(a, f) \circ \eta^P_a
    \colon \idfunc_{\catc} \nattr T_P(f)$
    for $a, b \in \cats$ and $f \colon a \to b$ in $\cats$.
  \end{itemize}
  Then $(T_P, \eta^{T_P}, \mu^{T_P}, \gu^{T_P})$ is
  a $\cats^{\nabla}$-graded monad with a generalised unit.
\end{proposition}

\begin{proposition}[Parameterised monads from category-graded monads with generalised units]
  Let $T \colon (\cats^{\nabla})^{\opposite} \to \Endo{\catc}$
  with ${(\gu^T_{f} \colon \idfunc_{\catc} \nattr T_f)}_{f \in \cats}$
  be a category-graded monad with generalised unit
  satisfying $T_{\inj_1 \idmorph_a} = T_{\inj_2 (a,a)}$.
  \footnote{
    The condition $T_{\inj_1 \idmorph_a} = T_{\inj_2 (a,a)}$ is
    necessary for the construction of $P_T(f,g)$ to be well-defined,
    but it was not imposed in \cite{OWE20}.
  }
  We define a functor
  $P_T \colon \cats^{\opposite} \times \cats \to [\catc, \catc]$
  and natural transformations $\eta^{P_T}$ and $\mu^{P_T}$
  as follows.
  \begin{itemize}
  \item $P_T(a,b) \defeq T_{\inj_2 (a,b)}$ and $P_T(a,a) \defeq T_{\inj_1 \idmorph_a} = T_{\inj_2 (a,a)}$ for objects.
  \item For $a, a', b, b' \in \ob{\cats}$ and $f \colon a' \to a$ and $g \colon b \to b'$, $P_T(f,g)$ is
    \[
    \begin{tikzcd}
      P_T(a,b) = T_{\inj_2 (a,b)}
      \ar[r, Rightarrow, "{\gu^T_{f} T_{\inj_2 (a,b)}}"{yshift=3pt}]
      &
      T_{\inj_1 f} T_{\inj_2 (a,b)}
      \ar[r, Rightarrow, "{T_{\inj_1 f} T_{\inj_2 (a,b)} \gu^T_g}"{yshift=3pt}]
      &
      T_{\inj_1 f} T_{\inj_2 (a,b)} T_{\inj_1 g}
      \ar[r, Rightarrow, "{\mu^T_{\inj_1 f, \inj_2 (a,b)} T_{\inj_1 g}}"{yshift=3pt}]
      &
      T_{\inj_2 (a',b)} T_{\inj_1 g}
      \ar[d, Rightarrow, "{\mu^T_{\inj_2 (a',b), \inj_1 g}}"]
      \\
      & & &
      T_{\inj_2 (a',b')} = P_T(a', b').
    \end{tikzcd}
    \]
  \item $\eta^{P_T}_a \defeq \eta^T_a \colon \idfunc_{\catc} \nattr P(a,a)$,
    and $\mu^{P_T}_{a,b,c} \defeq \mu^T_{\inj_2 (a,b), \inj_2 (b,c)} \colon P(a,b) P(b,c) \nattr P(a,c)$.
  \end{itemize}
  Then $(P_T, \eta^{P_T}, \mu^{P_T})$ is a parameterised monad.
\end{proposition}

\begin{theorem}
  Parameterised monads correspond to category-graded monads with generalised units.
  More precisely,
  there is a one to one correspondence between
  the set of $\cats$-parameterised monads and
  the set of $\cats^{\nabla}$-graded monads with generalised units $(T, \gu)$
  satisfying $T_{\inj_1 f} = T_{\inj_2 (a,b)}$
  for all $a, b \in \ob{\cats}$ and $f \colon a \to b$ in $\cats$.
\end{theorem}

\subsection{Correspondence of Eilenberg-Moore Constructions}
Let $P \colon \cats^{\opposite} \times \cats \to [\catc, \catc]$
be a parameterised monad.
We obtain the functor
$\EM{T^P} \colon {(\cats^{\nabla})}^{\opposite} \to \cat$
by Eilenberg-Moore construction on
$T_P \colon {(\cats^{\nabla})}^{\opposite} \laxto \Endo{\catc}$.
In this section,
we describe the relation between $P$-algebras and $T_P$-algebras.
For $a \in \ob{\cats}$,
Recall that the category of $T^P$-algebra at $a$, $\EM{T^P}(a)$,
consists of the following data.
\begin{itemize}
\item Objects are $(A,h)$ where
  for $f \colon a \to b$, $g \colon b \to c$ in $\cats^{\nabla}$
  $A \colon \sum_{b} {(\cats^{\nabla})}(b,a) \to \catc$
  and
  $\xi_{f,g} \colon T_f A_g \to A_{g \circ f}$
  are compatible with $\eta^{T_P}$ and $\mu^{T_P}$.
\item A morphism from $(A, \xi)$ to $(A', \xi')$
  is a natural transformation $\alpha \colon A \nattr A'$
  which compatible with $\xi$ and $\xi'$.
\end{itemize}

On the other hand,
the category of $P$-algebra $\catc^P$ consists of the following data \cite{Atkey09}.
\begin{itemize}
\item Objects are $(B, \zeta)$ where
  $B \colon \cats^{\opposite} \to \catc$
  and
  $\zeta_{a,b} \colon P(a,b) B_b \to B_a$
  are compatible with $\eta^P$ and $\mu^P$.
\item A morphism from $(B, \zeta)$ to $(B', \zeta')$
  is a natural transformation $\beta \colon B \nattr B'$
  which compatible with $\zeta$ and $\zeta'$.
\end{itemize}

\textbf{Construction of $T_P$-algebras from $P$-algebras.}
Given a $P$-algebra $(B, \zeta)$,
we can construct a $T_P$-algebra $(A, \xi)$ at $a$ for any $a \in \ob{\cats}$ by
$A_{g} \defeq B_b$ and $\xi_{f,g} \defeq \zeta_{c,b} \colon {T_P}_{f} A_{g} = P(c,b) B_b \to B_c = A_{g \circ f}$
where $f \colon c \to b$ and $g \colon b \to a$ are morphisms in $\cats^{\nabla}$.
This $T_P$-algebra $(A,\xi)$ satisfies
$A_{\inj_1 f} = A_{\inj_2 (b,a)}$
for all $b \in \ob{\cats}$ and $f \colon b \to a$ in $\cats$.

Next,
we consider a construction of morphisms of $T_P$-algebras
from morphisms of $P$-algebras.
Consider a morphism $\beta \colon (B, \zeta) \to (B', \zeta')$ between $P$-algebras.
Recall that $\beta$ is a natural transformation $\beta \colon B \nattr B'$
with $\zeta'_{b, c} \circ (P(b,c) \beta_c) = \beta_b \circ \zeta_{b,c}$.
Let $(A, \xi)$ and $(A', \xi')$ be
the $T_P$-algebras constructed from $(B, \zeta)$ and $(B', \zeta')$, respectively.
We define a natural transformation $\alpha \colon A \nattr A'$ to be
$\alpha_{f \colon b \to a} \defeq \beta_b \colon A_f = B_b \to B'_b = A'_f$.
This $\alpha$ becomes a $T_P$-homomorphism at $a$.
Indeed, for all $f \colon b \to c$ and $g \colon c \to a$ in $\cats^{\nabla}$
the following commutative diagram holds.
\[ \begin{tikzcd}
  {T_P}_f A_g \ar[r, "{{T_P}_f \alpha_f}"] \ar[d, equal] \ar[ddd, bend right=50, "h_{f,g}"']&
  {T_P}_f A'_g \ar[d, equal] \ar[ddd, bend left=50, "h'_{f,g}"] \\
  P(b,c)B_c \ar[d, "k_{b,c}"] \ar[r, "{P(b,c)\beta_c}"] &
  P(b,c)B'_c \ar[d, "k'_{b,c}"] \\
  B_b \ar[r, "{\beta_b}"] \ar[d, equal]&
  B'_b \ar[d, equal] \\
  A_{gf} \ar[r, "{\alpha_{gf}}"] &
  A'_{gf} \\
\end{tikzcd} \]

Summarizing the above discussion,
we obtain a functor
$F_a \colon \catc^P \to \EM{T_P}(a)$
where $a \in \ob{\cats}$.
In particular,
we get
$F_a \colon \catc^P \to \catc^{T_P}_{a}$
where $\catc^{T_P}_{a}$ be the full subcategory of $\EM{T_P}(a)$
whose objects are $(A,\xi)$ satisfying 
$A_{\inj_1 f} = A_{\inj_2 (b,a)}$
for all $b \in \ob{\cats}$ and $f \colon b \to a$ in $\cats$.

\textbf{Construction of $P$-algebras from $T_P$-algebras.}
Conversely, given a $T_P$-algebra $(A, \xi)$ at $a$ satisfying
$A_{\inj_1 f} = A_{\inj_2 (b,a)}$
for all $b \in \ob{\cats}$ and $f \colon b \to a$ in $\cats$,
we can construct a $P$-algebra $(B, \zeta)$ as follows.
\begin{gather*}
  B_b \defeq A_{\inj_2 (b,a)},
  \quad
  B_f \defeq
  \left( \begin{tikzcd}[ampersand replacement = \&]
    B_c = A_{\inj_2 (c, a)} \ar[r, "{\gu^{T_P}_f}"] \&
    {T_P}_f A_{\inj_2 (c, a)} \ar[r, "\xi_{f, \inj_2 (c, a)}"] \&
    A_{\inj_2 (c,a) \circ f} = B_b
  \end{tikzcd} \right),
  \\
  \zeta_{b,c} \defeq \xi_{\inj_2 (b,c), \inj_2 (c,a)}
  \colon P(b,c) B_c = {T_P}_{\inj_2 (b,c)} A_{\inj_2 (c,a)} \to A_{\inj_2 (b, a)} = B_b.
\end{gather*}
We prove that $B$ is a functor.
Firstly, we have
$ B_{\idmorph_b}
= \xi_{\idmorph_b, \inj_2 (b,a)} \circ \gu^{T_P}_{\idmorph_b} A_{\inj_2 (b,a)}
= \xi_{\idmorph_b, \inj_2 (b,a)} \circ \eta^{T_P}_{b} A_{\inj_2 (b,a)}
= \idmorph_{A_{\inj_2 (b,a)}}.$
Secondly, we have
\[
  B_f \circ B_g
  =
  \left(
  \begin{tikzcd}
    A_{\inj_2 (d,a)}
    \ar[r, "{\gu^{T_P}_g A_{\inj_2 (d,a)}}"{yshift=3pt}]
    & {T_P}_g A_{\inj_2 (d,a)}
    \ar[r, "h_{g, \inj_2 (d,a)}"]
    & A_{\inj_2 (d,a) \circ g}
    \ar[r, "{\gu^{T_P}_f A_{\inj_2 (d,a) \circ g}}"{yshift=3pt}]
    & {T_P}_f A_{\inj_2 (d,a) \circ g}
    \ar[r, "\xi_{f, \inj_2 (d,a) \circ g}"]
    & A_{\inj_2 (d,a) \circ g \circ f}
  \end{tikzcd}
  \right)
\]
and
\[
B_{g \circ f} =
\left(
\begin{tikzcd}
  A_{\inj_2 (d,a)} \ar[r, "{\gu^{T_P}_{gf} A_{\inj_2 (d,a)}}"{yshift=3pt}] &
  {T_P}_{g \circ f} A_{\inj_2 (d,a)} \ar[r, "\xi_{g \circ f, \inj_2 (d,a)}"] &
  A_{\inj_2 (d,a) \circ g \circ f}
\end{tikzcd}
\right)
\]
for $f \colon b \to c$ and $g \colon c \to d$ in $\cats$.
These two morphisms are equal
because the following diagram commutes.
\[
\begin{tikzcd}
  A_{\inj_2 (d,a)}
  \ar[d, "{\gu^{T_P}_g A_{\inj_2 (d,a)}}"']
  \ar[rrd, bend left=15, "{\gu^{T_P}_{g \circ f} A_{\inj_2 (d,a)}}"]
  & & \\
  {T_P}_g A_{\inj_2 (d,a)}
  \ar[r, "{\gu^{T_P}_f T_g A_{\inj_2 (d,a)}}"{yshift=3pt}]
  \ar[d, "\xi_{g, \inj_2 (d,a)}"']
  & {T_P}_f {T_P}_g A_{\inj_2 (d,a)}
  \ar[r, "\mu^{T_P}_{f,g} A_{\inj_2 (d,a)}"{yshift=3pt}]
  \ar[d, "T_f \xi_{g, \inj_2 (d,a)}"]
  & {T_P}_{g \circ f} A_{\inj_2 (d,a)}
  \ar[d, "\xi_{g \circ f, \inj_2 (d,a)}"] \\
  A_{\inj_2 (d,a) \circ g}
  \ar[r, "{\gu^{T_P}_f A_{\inj_2 (d,a) \circ g}}"'{yshift=-3pt}]
  & {T_P}_f A_{\inj_2 (d,a) \circ g}
  \ar[r, "\xi_{f, \inj_2 (d,a) \circ g}"'{yshift=-3pt}]
  & A_{\inj_2 (d,a) \circ g \circ f}
\end{tikzcd}
\]
In the above diagram,
the top triangle is commutative by the definition of generalised unit,
the left bottom square is commutative by the naturality of $\gu^{T_P}$
and the right bottom square is commutative the definition of $T_P$-algebra $(A,\xi)$.
Summarizing the above discussion,
we obtain a functor
$G_a \colon \catc^{T_P}_a \to \catc^P$.
We can easily check the following theorem by the definition of $F_a$ and $G_a$.
\begin{theorem}
  The category $\catc^P$ of $P$-algebras
  is isomorphic to the category $\catc^{T_P}_a$ of $T_P$-algebras at $a$:
  $F_a : \catc^P \cong \catc^{T_P}_a : G_a$.
\end{theorem}

\end{document}